\documentclass{article}

\usepackage{amssymb, amsmath, amsthm}

\usepackage{fullpage}
\usepackage[english]{babel} 
\usepackage[utf8]{inputenc} 
\usepackage{graphicx}
\usepackage{epigraph}

\usepackage{thmtools} 
\usepackage{thm-restate}
\usepackage{xcolor}
\usepackage{hyperref}
\hypersetup{
        colorlinks,
    linkcolor={red!50!black},
    citecolor={blue!50!black},
    urlcolor={blue!80!black}
}

\usepackage[ 
backend=bibtex,maxbibnames=99,sortcites]{biblatex}
\usepackage[font={small}]{caption}

\newcommand{\er}{\exists\mathbb{R}}
\newcommand{\np}{\mathsf{NP}}
\newcommand{\pspace}{\mathsf{PSPACE}}

\newcommand{\R}{\mathbb{R}}

\newcommand{\ER}{\exists\R}

\newcommand{\II}{\mathcal I}

\newcommand{\etr}{{\sc{Etr}}}
\newcommand{\etrinv}{{\sc{Etr-Inv}}}
\newcommand{\rangeetrinv}{{\sc{Range-Etr-Inv}}}
\newcommand{\mccf}{{\sc{Minimum Convex Cover}}}
\newcommand{\mrccf}{{\sc{Minimum Restricted Convex Cover}}}
\newcommand{\mcc}{{\sc{MCC}}}
\newcommand{\mrcc}{{\sc{MRCC}}}
\newcommand{\rpcf}{{\sc{Rectilinear Picture Compression}}}
\newcommand{\mtcf}{{\sc{Minimum Triangle Cover}}}
\newcommand{\agp}{{\sc{Art Gallery}}}
\newcommand{\poly}{\mathcal P}
\newcommand{\pocket}{P}
\newcommand{\piece}{Q}
\newcommand{\cover}{\mathcal{\piece}}
\newcommand{\proj}[1]{\widehat{#1}}
\newcommand{\tilt}{\eta}
\newcommand{\triangleC}{\Delta}
\newcommand{\barr}{\Pi}
\newcommand{\pp}[2]{\mathbf{p}(#1,#2)} 
\newcommand{\ppCover}[3]{\mathbf{p}_{#3}(#1,#2)}
\newcommand{\val}[2]{\mathbf{v}(#1,#2)} 
\newcommand{\valCover}[3]{\mathbf{v}_{#3}(#1,#2)}
\newcommand{\xx}[2]{\mathbf{x}(#1,#2)} 
\newcommand{\Phiineq}{\Phi_{\text{ineq}}}
\newcommand{\Xineq}{X_{\text{ineq}}}

\newcommand{\eps}{\varepsilon}
\newcommand{\mydef}{:=}

\newtheorem{theorem}{Theorem}
\newtheorem{lemma}[theorem]{Lemma}
\newtheorem{observation}[theorem]{Observation}

\newtheorem{definition}[theorem]{Definition}

\title{Covering Polygons is Even Harder}
\author{%
	Mikkel Abrahamsen%
	\footnote{%
		Basic Algorithms Research Copenhagen (BARC), University of Copenhagen.
		BARC is supported by the VILLUM Foundation grant 16582.%
	}%
}

\date{June 4, 2021}

\begin{document}
\maketitle

\begin{abstract}
In the \mccf\ (\mcc) problem, we are given a simple polygon $\poly$ and an integer $k$, and the question is if there exist $k$ convex polygons whose union is $\poly$.
It is known that \mcc\ is $\np$-hard [Culberson \& Reckhow: Covering polygons is hard, FOCS 1988/Journal of Algorithms 1994] and in $\ER$ [O'Rourke: The complexity of computing minimum convex covers for polygons, Allerton 1982].
We prove that \mcc\ is $\ER$-hard, and the problem is thus $\ER$-complete.
In other words, the problem is equivalent to deciding whether a system of polynomial equations and inequalities with integer coefficients has a real solution.

If a cover for our constructed polygon exists, then so does a cover consisting entirely of triangles.
As a byproduct, we therefore also establish that it is $\ER$-complete to decide whether $k$ triangles cover a given polygon.

The issue that it was not known if finding a minimum cover is in $\np$ has repeatedly been raised in the literature, and it was mentioned as a ``long-standing open question'' already in 2001 [Eidenbenz \& Widmayer: An approximation algorithm for minimum convex cover with logarithmic performance guarantee, ESA 2001/SIAM Journal on Computing 2003].
We prove that assuming the widespread belief that $\np\neq\ER$, the problem is not in $\np$.

An implication of the result is that many natural approaches to finding small covers are bound to give suboptimal solutions in some cases, since irrational coordinates of arbitrarily high algebraic degree can be needed for the corners of the pieces in an optimal solution.
\end{abstract}

\epigraph{\textit{
There is no point in studying a mathematical model unless it helps to solve certain practical problems.}}{--- Theodosios Pavlidis~\cite{PAVLIDIS19725}}

\epigraph{\textit{It is certainly counterproductive to dismiss complicated theoretical methods solely on the grounds of impracticality.}}{--- Bernard Chazelle~\cite{chazelle1985approximation}}

\thispagestyle{empty}

\newpage

\pagenumbering{arabic} 

\section{Introduction}

Polygons are among the geometric structures that are most frequently used to model physical objects, as they are suitable for representing a wide variety of shapes and figures in computer graphics and vision, pattern recognition, robotics, computer-aided design and manufacturing, and other computational fields.
Polygons may have very complicated shapes that make it difficult to find algorithms to process them directly.
A natural first step in designing algorithms is to decompose the given polygon $\poly$ into more basic \emph{pieces} of a restricted type that permits very efficient processing.
Here, the union of the pieces must be exactly the given polygon $\poly$.
When such a decomposition has been obtained, the partial solutions to the individual pieces can be combined to obtain a solution for the complete polygon $\poly$.
By ``more basic pieces,'' we mean pieces that belong to a more restricted class of polygons.
Ideally, we find a decomposition of the polygon $\poly$ into the minimum number of pieces.

Many different decomposition problems arise this way, depending on restrictions on the polygon $\poly$, the type of basic pieces, and whether the decomposition is a cover or a partition.
In covering problems, we just require the union of the pieces to equal $\poly$, whereas in partition problems, we have the further requirement that the pieces be interior-disjoint.
Another important distinction is whether or not \emph{Steiner points} are allowed.
A Steiner point is a corner of a piece in a decomposition which is not also a corner of $\poly$.
Furthermore, it often makes a big difference whether or not $\poly$ is allowed to have holes.

Because of the many variants of decomposition problems, their applicability within many practical domains, 
and the very appealing fundamental nature and the creativity and technical skills required to solve them, numerous papers have been written about these problems.
The vast literature is documented in several highly-cited books and survey papers that give an overview of the state-of-the-art at the time of publication~\cite{chazelle1985approximation,keil1985minimum,shermer1992recent,chazelle1994decomposition,keil1999polygon,o1987art,o2004polygons}. 

One of the first decomposition problems to be studied was that of covering a polygon with convex polygons.
Pavlidis studied this problem from a practical angle in relation to shape analysis and pattern recognition in a series of papers and a book, the first paper from 1968~\cite{PAVLIDIS1968165,PAVLIDIS1972421,PAVLIDIS19725,feng1975decomposition,pavlidis1977structural,PAVLIDIS1978243}; see Figure~\ref{fig:pavlidis} (left).
The related $\np$-complete problem of covering an orthogonal polygon with a given number of rectangles was mentioned by Garey and Johnson's famous book~\cite[p.~232]{garey1979computers}, who denoted the problem \rpcf, since a collection of pixels can be compactly represented as their minimum rectangle cover.
For more recent practical work involving versions of convex covering problems, see the surveys~\cite{rodrigues2018part,suk2021decomposition} and the paper~\cite{asafi2013weak}.

\begin{figure}
\centering
\includegraphics[page=11]{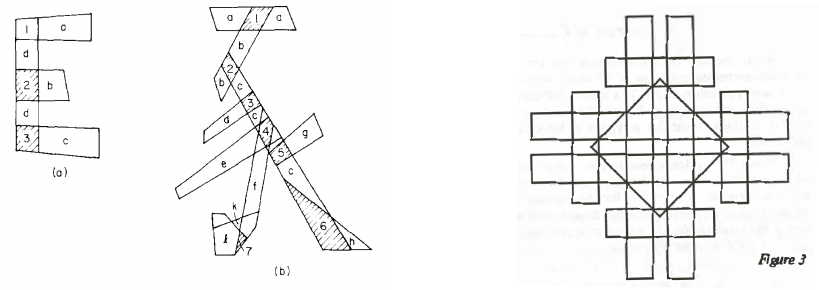}
\caption{Left: Two polygons (an E and a Chinese character for `pig') with convex covers drawn on top, from the paper~\cite{PAVLIDIS1978243}.
Right: A polygon where the minimum cover has a piece not bounded by edge extensions and needing Steiner points, from the paper~\cite{o1982complexity}.
}
\label{fig:pavlidis}
\end{figure}

O'Rourke and Supowit~\cite{o1983some} proved that for polygons with holes, it is $\np$-hard to find minimum covers using convex, star-shaped, and spiral polygons as pieces.
This was established with and without Steiner points allowed.
When Steiner points are not allowed, the covering problems are clearly in $\np$ and are thus $\np$-complete, whereas the ones with Steiner points allowed are not immediately seen to be decidable at all.
In the paper~\cite{o1982complexity}, O'Rourke proved the convex covering problem with Steiner points to be decidable by showing that any given instance can be expressed as a Tarski formula.
In modern terms, we can say that he proved $\ER$-membership.
He gave the example shown in Figure~\ref{fig:pavlidis} (right) of a polygon $\poly$ where Steiner points are needed and edge extensions of $\poly$ are not sufficient to form the pieces of a minimum cover.
The figure is now the logo of \href{https://www.computational-geometry.org/}{The Society of Computational Geometry} and often also of \href{https://cse.buffalo.edu/socg21/}{Symposium on Computational Geometry}.

Since the $\np$-hardness reductions in~\cite{o1983some} relied on polygons with holes, it was still not known if the covering problems for polygons without holes could be solved efficiently.
Chazelle and Dobkin~\cite{chazelle1985optimal} showed already in 1979 that a polygon without holes can be \emph{partitioned} into a minimum number of convex pieces in polynomial time; a problem that had also been believed to be $\np$-hard, so this might have given hope for the covering problems as well.
However, it was soon proved that some covering problems are $\np$-hard even for polygons without holes.
The first one was apparently the problem of covering a polygon with a minimum number of star-shaped polygons, which is usually known as the \agp\ problem.
The first proof is often attributed to Alok Aggarwal's PhD thesis~\cite{aggarwal1984art} (see for instance~\cite{culberson1994covering}), which is unfortunately practically unavailable.
The first published proof appears to be in a paper by Lee and Lin~\cite{DBLP:journals/tit/LeeL86}.
Then followed the proof by Culberson and Reckhow~\cite{culberson1994covering} that it is likewise $\np$-hard to cover polygons without holes with a minimum number of convex pieces.
The authors added the comment: \emph{``We are unable to show that general convex covering is in $\np$''}.
The issue that the problem is not known to be in $\np$ has been raised in many other papers and books~\cite{o1982complexity,JOHNSON1982182,eidenbenz2003approximation,shermer1992recent,christ2011beyond,o1987art,keil1999polygon,keil1985minimum}, and was mentioned as a ``long-standing open question'' already in 2001~\cite{eidenbenz2003approximation}.
Christ~\cite{christ2011beyond} proved that deciding if a polygon can be covered by a minimum number of \emph{triangles} is also $\np$-hard.

In this paper, we prove that it is $\ER$-complete to decide if a polygon can be covered by a given number of convex pieces.
The problem is thus not in $\np$, assuming the widespread belief that $\np\neq\ER$.
Our reduction uses techniques that were developed for proving that \agp\ and some versions of geometric packing are $\ER$-complete~\cite{abrahamsen2018art,abrahamsen2020framework}.
As a biproduct, we show that it is even $\ER$-complete to find a minimum triangle cover.
The hardness holds also when the corners of $\poly$ are in general position, i.e., no three are collinear.

\subsection{Existential theory of the reals}

In order to define the complexity class $\ER$, we first define the problem \etr\ in the style of Garey and Johnson~\cite{garey1979computers}.

\noindent
\textbf{Instance}:
A well-formed formula $\Phi(x_1,x_2,\ldots,x_n)$ using symbols from the set
\[\left\{x_1,x_2,\ldots, x_n, \land,\lor,\lnot, 0 ,1 ,+ ,- ,\cdot,\allowbreak\ (\ ,\ )\ ,=,<,\leq\right\}.\]

\noindent
\textbf{Question}:
Is the expression
\[\exists x_1, x_2, \ldots, x_n\in\mathbb R \colon \Phi(x_1,x_2,\ldots,x_n)\]
true? \\

The complexity class $\ER$ consists of all problems that are many-one reducible to \etr\ in polynomial time, and a problem is $\ER$-hard if there is a reduction in the other direction.
It is currently known that \[ \np \subseteq \ER \subseteq \pspace.\]
It is not hard see that the problem \etr\ is $\np$-hard, yielding the first inclusion.
The containment $\ER \subseteq \pspace$ is highly non-trivial, and it was first established by Canny~\cite{canny1988algebraic}.

As examples of $\ER$-complete problems, we mention problems related to realization of order-types~\cite{richter1995realization, mnev1988universality, shor1991stretchability},
graph drawing~\cite{bienstock1991some, AreasKleist, AnnaPreparation},
recognition of geometric graphs~\cite{cardinal2017intersection, cardinal2017recognition,kang2011sphere,mcdiarmid2013integer},
straightening of curves~\cite{erickson2019optimal}, guarding polygons~\cite{abrahamsen2018art},  Nash-equilibria~\cite{berthelsen2019computational,garg2015etr}, linkages~\cite{abel, schaefer2013realizability, Schaefer-ETR}, matrix-decompositions~\cite{NestedPolytopesER,shitov2016universality,Shitov16a}, polytope theory~\cite{richter1995realization}, and geometric packing~\cite{abrahamsen2020framework}.
See also the surveys~\cite{CardinalSurvey, matousek2014intersection, Schaefer2010}.

\subsection{Results}

Before stating our result, let us define the covering problem in more detail.
We define a \emph{polygon} $\poly$ to be a compact region in the plane such that the boundary $\partial\poly$ is a closed, simple curve consisting of finitely many line segments.
We now define the problem \mccf\ (\mcc\ for brevity) as follows.

\noindent
\textbf{Instance}:
A polygon $\poly$ represented as an array of the coordinates of the corners in cyclic order, and a positive integer $k$.
The corners have rational coordinates.

\noindent
\textbf{Question}:
Do there exist $k$ convex polygons $\piece_1,\ldots,\piece_k$ such that $\bigcup_{i=1}^k \piece_i=\poly$? \\

We get the problem \mtcf\ by requiring that each piece $\piece_i$ is a triangle.
We can now state the main result of the paper.

\begin{restatable}{theorem}{mainThm}\label{thm:mainThm}
\mccf\ and \mtcf\ are $\ER$-complete.
\end{restatable}

Recall that O'Rourke~\cite{o1982complexity} proved $\ER$-membership already in 1982.
Alternatively, $\ER$-membership can be easily proven using the recent framework by Erickson, van der Hoog, and Miltzow~\cite{erickson2020smoothing}.
This paper is therefore only concerned with proving $\ER$-hardness.

An implication of the result is that many natural approaches to finding small covers are bound to give suboptimal solutions in some cases.
One might attempt to make covers consisting of pieces with corners chosen from some discrete set of points inside the given polygon.
For instance, Eidenbenz and Widmayer~\cite{eidenbenz2003approximation} gave a $O(\log n)$-factor approximation algorithm for finding a minimum convex cover for a polygon $\poly$ by choosing pieces with corners from a set of $O(n^{16})$ points.
These points are obtained by first making the arrangement of lines through all pairs of corners of $\poly$.
We then construct all lines through pairs of intersection points in the first arrangement.
The intersection points of this final arrangement, of which there are $O(n^{16})$, are the candidate set of corners of the pieces from which a cover is constructed.

Eidenbenz and Widmayer showed that for this particular set, there exists a cover consisting of at most $3$ times more pieces than in the unrestricted optimum.
Our result shows that the optimal solution cannot always be found by choosing the corners from such a set of points, even if the process of making new lines and intersection points was repeated any finite number of times.
It follows from our construction that there are polygons where the corners of some pieces in any optimal cover have irrational coordinates with arbitrarily high algebraic degree.
As an example, we can consider an equation such as $x^5-4x+2=0$~\cite{bastida_lyndon_1984}, which has one real solution and that solution cannot be expressed by radicals (that is, the solution cannot be written using integers and basic arithmetic operations including powers and roots).
We can then transform the equation to an instance of \mcc, so that in the optimal cover, some pieces have coordinates that cannot be expressible by radicals.
We state this as a corollary to our construction, to be proven in Section~\ref{sec:putting}.

\begin{restatable}{corollary}{corr}\label{corr:corr}
There exists a polygon $\poly$ such that in any minimum convex cover for $\poly$, a piece has a corner with a coordinate that is not expressible by radicals.
\end{restatable}

\subsection{Structure of the paper and basic techniques}
We prove that \etr\ reduces to \mcc\ using two intermediate problems: \rangeetrinv\ and \mrcc, which we define in Section~\ref{sec:aux}.
Essentially, an instance of \rangeetrinv\ is a conjunction of addition constrains of the form $x+y=z$ and inversion constraints of the form $x\cdot y=1$.
Each variable is restricted to a tiny subinterval of $[\frac 12,2]$, and the goal is to decide if there exist values of the variables that satisfy all the addition and inversion constraints.
This problem was developed recently in order to prove $\ER$-hardness of geometric packing problems~\cite{abrahamsen2020framework}.

The problem \mrcc\ is a technical version of \mcc\ where only some crucial parts of the polygon have to be covered, namely a specific set of \emph{marked corners} and a set of \emph{marked rectangles}.
The introduction of the problem \mrcc\ is crucial in order to keep our reductions manageable.
In Section~\ref{sec:aux}, we define the problem and give a reduction from \mrcc\ to \mcc.
In the proof, we show that everything else than the marked corners and rectangles can be covered in a generic way by adding some spikes to the boundary of the polygon.
We therefore get an instance of \mcc\ which is equivalent to the instance of \mrcc\ that we are reducing from.

In Section~\ref{sec:etrinv-to-mrcc}, we give the reduction from \rangeetrinv\ to \mrcc, which is the major part of the work.
We do reductions in reversed order so that the reader can get a better understanding of the problem \mrcc, and the rationale for introducing it, before seeing the rather long reduction from \rangeetrinv\ to \mrcc.
In Section~\ref{sec:concluding}, we conclude the paper by outlining how triples of collinear corners can be avoided in our construction and mentioning an open problem.

\begin{figure}
\centering
\includegraphics[page=29]{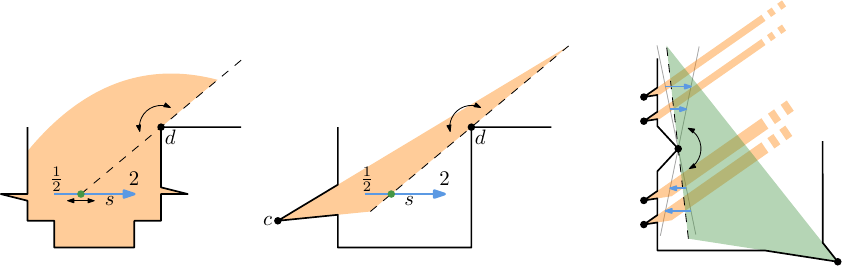}
\caption{Representing variables in the paper~\cite{abrahamsen2018art} and this paper.
Left: A guard segment $s$ with a guard representing the value $1$.
When the guard moves to the right, more visibility is blocked due to the corner $d$.
Middle: A piece covering the marked corner $c$ represents the value $1$ on the variable segment $s$.
By rotating the dashed edge around the corner $d$, more or less of $s$ will be covered.
Right: Several variable segments can be partially covered by a single piece and then ``copied'' into different gadgets by the orange propagation pieces.
}
\label{fig:artGalleryComp}
\end{figure}

%
One of our basic tools in the reduction to \mrcc\ is to represent variables using specific horizontal \emph{variable segments} contained in the constructed polygon; see Figure~\ref{fig:artGalleryComp}.
A segment $s$ representing a variable $x$ corresponds to the interval $[1/2,2]$ of values that $x$ can attain, with the endpoints representing $x=1/2$ and $x=2$.
The values of points in between the endpoints are defined by linear interpolation.
A part of a segment $s$ must be covered by a piece $\piece$ that also covers a particular marked corner of the polygon.
The right-most (or left-most) point in the intersection $s\cap \piece$ defines the value that $\piece$ represents on $s$.
The rest of $s$ must then be covered by another piece (which also covers another marked corner).

In the bottom of the polygon, we have some \emph{base pockets} containing variable segments representing all the variables.
We describe gadgets for addition inequalities, $x+y\geq z$ and $x+y\leq z$, and inversion inequalities, $x\cdot y\geq 1$ and $x\cdot y\leq 1$, and these gadgets are placed far to the right and above the base pockets.
The gadgets also contain variable segments.
We create some corridors that connect the gadgets and the base pockets, and they ensure that if a cover exists, then appropriate inequalities hold between the values represented in the base pockets and those in the gadgets.
We are therefore able to conclude that if a cover exists, then so does a solution to the instance of \rangeetrinv\ that we are reducing from.

Our reduction shares some resemblance with that from~\cite{abrahamsen2018art}, and many of the geometric tools underlying our construction are similar to the ones used in that paper.
In~\cite{abrahamsen2018art}, variables are also represented by segments in the polygon, but the actual value is defined by the position of a guard standing on the segment (as shown in Figure~\ref{fig:artGalleryComp}), whereas in our case, it is the left- or right-most point covered by a convex piece.
The geometric principles underlying the addition and inversion gadgets are the same in the two papers, but again the actual realizations of the gadgets are very different.

A key insight is that one convex piece in a base pocket can partially cover any number of variable segments, and thus represent numerous copies of the same variable, as sketched in Figure~\ref{fig:artGalleryComp} (right).
Each of these segments will have its own \emph{propagation corner}, which is a marked corner.
The propagation corner can then be covered by a \emph{propagation piece}, which must cover the remaining part of the segment.
Each propagation piece is very long and thin and sticks into a corridor far away.
This is needed so that a variable that appears in many addition and inversion constraints can be copied into all the corresponding gadgets.
In contrast to this, the reduction to \agp\ in~\cite{abrahamsen2018art} relied on the fact that a single guard can \emph{look} into several different corridors.
That technique is not possible to realize with convex pieces, because points in different corridors cannot see each other and hence no convex piece covers points in more than one corridor.


\section{Auxilliary problems: \rangeetrinv\ and \mrcc}\label{sec:aux}

In this section we introduce the two problems \rangeetrinv\ and \mrcc\ that will be used as intermediate problems in our reduction from \etr\ to \mcc.

\subsection{\rangeetrinv}

\begin{definition}
An \emph{\etrinv\ formula}~$\Phi=\Phi(x_1,\ldots,x_n)$ is a conjunction
\[
\left(\bigwedge_{i=1}^n 1/2\leq x_i\leq 2\right) \land \left(\bigwedge_{i=1}^m C_i\right),
\]
where $m\geq 0$ and each $C_i$ is of one of the forms
\begin{align*}
x+y=z,\quad x\cdot y=1 
\end{align*}
for $x,y,z \in \{x_1, \ldots, x_n\}$. 
\end{definition}

We can now define the problem \rangeetrinv\ as follows.

\noindent
\textbf{Instance:}
A tuple $\II=\left[\Phi,\delta,(I(x_1),\ldots,I(x_n))\right]$, where $\Phi$ is an \etrinv\ formula, $\delta>0$ is a (small) number, and, for each variable $x\in\{x_1,\ldots,x_n\}$, $I(x)$ is an interval $I(x)\subseteq [1/2,2]$ such that $|I(x)|\leq \delta$.
Define $V(\Phi) \mydef \{\mathbf x\in\R^n : \Phi(\mathbf x)\}$.
We are promised that $V(\Phi)\subset I(x_1)\times\cdots\times I(x_n)$.

\noindent
\textbf{Question:}
Is $V(\Phi)\neq\emptyset$? \\

The following theorem establishes that it suffices to make a reduction from \rangeetrinv\ in order to prove that \mcc\ is $\ER$-hard.

\begin{theorem}[\cite{abrahamsen2020framework}]
\rangeetrinv\ is $\ER$-complete, even when $\delta=O(n^{-c})$ for any constant $c>0$.
\end{theorem}

\subsection{\mrccf{} (\mrcc{})}\label{sec:mrcc}
We now turn our attention to the problem \mrccf{} (\mrcc{}).
As we will see later in this section, \mrcc\ can be reduced to \mcc.
Therefore, $\er$-hardness for \mrcc\ implies the same hardness for \mcc.
We define \mrcc\ as follows.

\noindent
\textbf{Instance:}
A tuple $\langle \poly, \mathcal{C}, \mathcal{R} \rangle$ consisting of the following parts:
\begin{itemize}
\item
A simple polygon $\poly$.

\item
A subset $\mathcal C = \{c_1, c_2, \ldots, c_{k}\}$ of the corners of $\poly$ called \emph{marked corners}.

\item
A set of pairwise disjoint axis-parallel rectangles $\mathcal R = \{ R_1, \ldots, R_m \}$ contained in $\poly$, called \emph{marked rectangles}.
\end{itemize}

\noindent
\textbf{Question:}
Do there exist $k=|\mathcal C|$ convex polygons $Q_1,\ldots,Q_{k}$ contained in $\poly$ such that $\mathcal C\subset \bigcup Q_i$ and $R\subset\bigcup Q_i$ for all marked rectangles $R\in\mathcal R$? \\

An instance of \mrcc{} comes with a promise which we will explain below.
In order to state the promise, we need some notation.
For a corner $c$ of a polygon $\poly$, define $\triangleC(c)$ to be the triangle with corners at $c$ and the two neighbouring corners of $\poly$.
Note that $\triangleC(c)$ may not be contained in $\poly$, but that will always be the case when we use the notation.

For a marked rectangle $R\in\mathcal R$, define the \emph{vertical bar} $\barr_v(R)$ to be the region of points that are vertically visible from $R$, i.e., the points $p\in\poly$ such that there is a point $q\in R$ where $pq$ is a vertical line segment contained in $\poly$.
In a similar way, we define the \emph{horizontal bar} $\barr_h(R)$ as the points that are horizontally visible from $R$; see Figure~\ref{fig:bars}.

In order to define the \emph{vertical trapezoidation} of a polygon $\poly$, we consider for each corner $c$ of $\poly$ the vertical segment that (i) contains $c$, (ii) is contained in $\poly$, and (iii) is maximal with respect to inclusion.
These segments partition $\poly$ into trapezoids, some of which are only triangles.
The \emph{horizontal trapezoidation} is defined similarly, using horizontal segments.

\begin{figure}
\centering
\includegraphics[page=12]{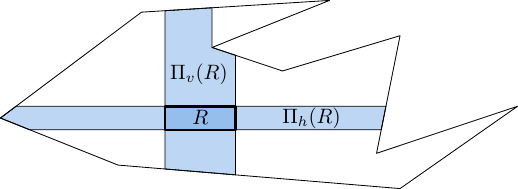}
\caption{A polygon with a marked rectangle $R$ the vertical and horizontal bars shown.}
\label{fig:bars}
\end{figure}

\paragraph{Promise.}
For the instances $\mathcal I= \langle \poly, \mathcal{C}, \mathcal{R} \rangle$ of \mrcc\ that we will consider, we promise that the following properties hold, which we call the \emph{\mrcc\ promise}.
For convenience, we have named each individual point of the promise.
\begin{enumerate}
\item
\emph{Skew triangle promise:}
For each marked corner $c$, the triangle $\triangleC(c)$ intersects the vertical and the horizontal line containing $c$ only at the point $c$.\label{promise:1}



\item
\emph{Trapezoid generality promise:}
In any trapezoid $T$ in the vertical trapezoidation of $\poly$, at least one right corner of $T$ is not a corner of $\poly$, unless the right corners of $T$ are coincident (so that $T$ is a triangle).
Similarly for the left corners.
Likewise, in any trapezoid in the horizontal trapezoidation, at least one top corner is not a corner of $\poly$ unless they are coincident, and similarly for the bottom corners. \label{promise:3}

\item
\emph{Bar intersection promise:}
For any two marked rectangles $R_i$ and $R_j$, the vertical bar of $R_i$ and
the horizontal bar of $R_j$ are either disjoint or their intersection is contained in a marked rectangle.\label{promise:2}

\item
\emph{Broad cover promise:}
If $\II$ has a cover, then there also exists a cover $\cover$ for $\II$ such that $\triangleC(c)\subset\bigcup_{\piece\in\cover} \piece$ for all $c\in\mathcal C$. \label{promise:4}
\end{enumerate}



\begin{lemma}\label{thm:mrcc-to-mcc}
Suppose that \mrccf\ is $\ER$-hard, even when restricted to instances satisfying the \mrcc\ promise.
Then \mccf\ is also $\ER$-hard.
\end{lemma}

\begin{proof}
Let $\mathcal I\mydef \langle\, \poly, \mathcal C, \mathcal R \,\rangle$
be an instance of \mrcc\ that satisfies the promise.
We first describe the high-level idea behind our reduction to \mcc.
We construct a larger polygon $\poly'\supset \poly$ by adding some number $s$
of spikes to $\poly$. 
A \emph{spike} is a triangle $S$ which is interior-disjoint from $\poly$ and such that the intersection $S\cap\poly$ is an edge of $S$, which we call the \emph{base} of the spike.
The corner of $S$ opposite of the base is called the \emph{tip} of the spike $S$.
The polygon \(\poly'\) will be defined as $\poly'\mydef\poly\cup\bigcup_{S\in\mathcal S} S$, where $\mathcal S$ is a set of $s$ spikes.
The spikes in $\mathcal S$ may overlap, but we choose them such that $\poly'$ has no holes.

We then get the instance $\mathcal I'\mydef \langle\, \poly', k'\,\rangle$ of \mcc, where the parameter \(k'\)  is defined as $k'\mydef k+s$ with $k\mydef|\mathcal C|$.
As we will see, our reduction implies that a cover for $\mathcal I'$ using $k'=k+s$ convex polygons contains $s$ polygons covering the added spikes.
The spikes have been chosen so that they cannot be covered while also covering $\mathcal R$ nor $\mathcal C$, so the remaining $k$ polygons must cover $\mathcal R$ and $\mathcal C$ and thus constitute a cover for $\mathcal I$.
On the other hand, given a cover $\cover$ for $\mathcal I$, one can add $s$ convex polygons contained in \(\poly'\) that cover the spikes and everything else of $\poly'$ not covered by $\cover$, and thus obtain a cover for $\mathcal I'$.
Therefore, the instances are equivalent.


To carry out this idea, we consider a vertical and a horizontal trapezoidation of (parts of) $\poly$ and in each trapezoid $T$, we add two or three spikes such that these spikes and $T$ can be covered with two or three triangles.
Let us explain the horizontal trapezoidation first.
Define $\mathcal T'_h$ as the horizontal trapezoidation of the polygons $\poly'_h\mydef \poly\setminus\bigcup \barr_h(R_j)$.
Consider a marked corner $c$.
Because of point~\ref{promise:1} of the promise of $\mathcal I$, $c$ is contained in a unique trapezoid of $\mathcal T'_h$, which is in fact a triangle $\triangleC'(c)\subseteq\triangleC(c)$.
We then define $\mathcal T_h\mydef \mathcal T'_h\setminus\{\triangleC'(c)\mid c\in\mathcal C\}$ as all the other trapezoids.
Likewise, we define $\poly_h\mydef\bigcup_{T\in\mathcal T_h} T$ as the region covered by the other trapezoids.

We now add \emph{horizontal} spikes to each trapezoid $T$ of $\mathcal T_h$ as shown in Figure~\ref{fig:spikes2}.
We must ensure that each added spike is interior-disjoint from $\poly$ to avoid ``collisions''.
This limits how long the spikes can be.
We say that a corner of $T$ is \emph{free} if it is not also a corner of $\poly$.
Let us first consider a non-degenerate trapezoid $T$, i.e., $T$ is not a triangle.
If a corner $c$ is not free, we should avoid to add a spike there, as the edge of $\poly$ incident at $c$ outside $T$ may cause a collision.
Because of point~\ref{promise:3} of the promise of $\mathcal I$, we know that a top corner and a bottom corner of $T$ are free.
If one is on the left side and the other on the right side of $T$ we can add two spikes, so that two triangles can cover the spikes and all of $T$, but nothing else of $\poly$.
If the top and bottom free corners are on the same side, we add two spikes on that side on one in the middle of the opposite side so that the three triangles can cover the three spikes and all of $T$, but nothing else.

If $T$ is a triangle, we add two spikes as shown, so that two triangles can cover the spikes and $T$, but nothing else.

\begin{figure}
\centering
\includegraphics[page=14]{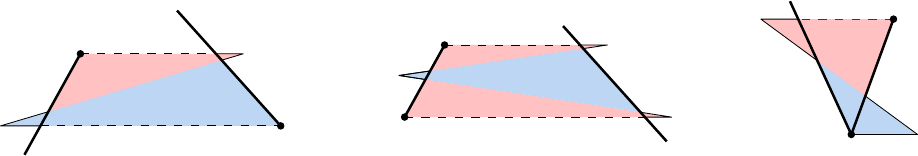}
\caption{%
The figure shows how we add spikes to trapezoids in the horizontal trapezoidation $\mathcal T_h$.
The fat edges are edges of $\poly$ and the dots show the corners of $\poly$, i.e., the corners that are not free.
Left: When there is a free left and right corner, we add two spikes such that $T$ and the two spikes can be covered by two triangles.
Middle: When there is no free left corner (as shown here) or no free right corner, we add three spikes.
Right: When the trapezoid is a triangle, we add two spikes as shown.
}
\label{fig:spikes2}
\end{figure}

A problem can occur when we want to make two horizontal spikes that share a horizontal side due to two neighbouring trapezoids $T_1$ and $T_2$ where $T_1$ is above $T_2$, as shown in Figure~\ref{fig:spikesMove}.
This results in a degenerate polygon with a corner with an interior angle of $2\pi$, or a non-functional merged spike, depending on the viewpoint.
We may assume without loss of generality that the bottom edge of $T_1$ is at least as long as the top edge of $T_2$.
We then move the spike of $T_2$ up and place it along the side of $T_1$.
The trapezoid $T_2$ and its spikes can now be covered with triangles that also cover a part of $T_1$, but nothing else of $\poly$.

\begin{figure}
\centering
\includegraphics[page=16]{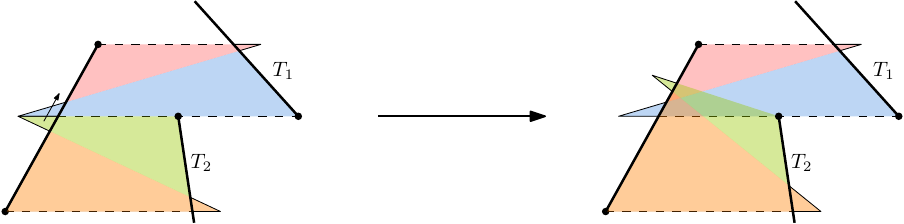}
\caption{%
Left: Two spikes share a horizontal edge and make a degenerate polygon.
Right: We move the spike of $T_2$ up and get two well-defined spikes.
}
\label{fig:spikesMove}
\end{figure}

We then consider the vertical trapezoidation $\mathcal T'_v$ of the polygons $\poly'_v\mydef \poly\setminus\bigcup \barr_v(R_j)$ and again define $\mathcal T_v$ by removing the triangles containing the marked corners from $\mathcal T_v$ and $\poly_v$ as the union of the trapezoids in $\mathcal T_v$.
We then add \emph{vertical} spikes to the trapezoids $\mathcal T_v$ in an analogous way.
Here, the horizontal and the vertical spikes may overlap, but we need to be careful that we do not introduce holes when we define $\poly'$ as the union of $\poly$ and all the spikes.
A hole is introduced if and only if a horizontal and a vertical spike overlap and the spikes have disjoint bases.
Such an overlap can be avoided by making the vertical spike sufficiently short so that it does not reach the horizontal spike; see Figure~\ref{fig:spikesNocross}.

\begin{figure}
\centering
\includegraphics[page=15]{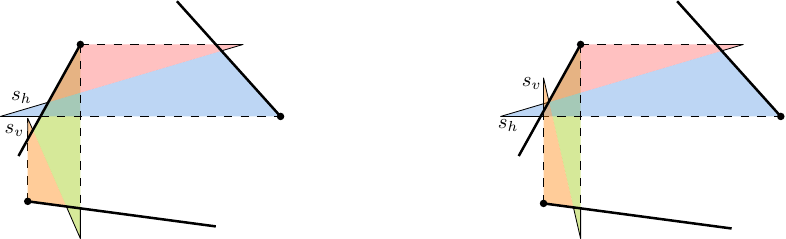}
\caption{%
Left: In order to avoid introducing a hole when we add spikes, we make the vertical spike $s_v$ so short that it does not intersect the horizontal spike $s_h$.
Right: Two spikes $s_v$ and $s_h$ overlap in an acceptable way without creating a hole.
}
\label{fig:spikesNocross}
\end{figure}

We now claim that
\begin{align}
\poly \setminus\left(\bigcup R_j\cup \bigcup \triangleC(c_i)\right)\subset \poly_v\cup \poly_h. \label{polyUnion}
\end{align}
To see this, we consider a point $p\in \poly \setminus(\bigcup R_j\cup \bigcup \triangleC(c_i))$.
If $p$ is not in the vertical or the horizontal bar of any marked rectangle, then $p$ is in both $\poly_v$ and $\poly_h$.
Otherwise, suppose that $p\in \barr_v(R)$ for some $R\in\mathcal R$.
Since $p$ is not in any marked rectangle, we get from point~\ref{promise:2} in the promise of $\mathcal I$ that $p$ is not in any horizontal bar of a marked rectangle.
Hence, $p\in \poly_h$.
Likewise, if $p\in\barr_h(R)$ for $R\in\mathcal R$, we get that $p\in\poly_v$.

It remains to verify that the instances $\mathcal I$ and $\mathcal I'\mydef \langle\poly',k'\rangle$ are equivalent.
Suppose first that $\mathcal I$ has a cover $\cover$ of size $k\mydef |\mathcal C|$.
By point~\ref{promise:4} of the promise of $\mathcal I$, we can assume that for each $c\in\mathcal C$, the whole triangle $\triangleC(c)$ is covered by $\cover$.
We define $\cover'$ to be the cover we get by adding to $\cover$ for each spike the maximal triangle in $\poly'$ that contains the spike.
Then $\cover'$ consists of $k'\mydef k+s$ polygons, where $s$ is the number of spikes, and we need to verify that $\bigcup_{\piece\in\cover'} Q=\poly'$.
Clearly, for every $\piece\in\cover'$, we have $\piece\subset \poly'$, so we just need to check that every point $p$ in $\poly'$ is in some polygon $\piece\in\cover'$.
If $p\in\poly'\setminus\poly$, then $p$ is in a spike and is thus covered.
If $p\in \Delta(c)$ for $c\in\mathcal C$ or $p\in R$ for $R\in\mathcal R$, then it follows from the properties of the cover $\cover$ of $\mathcal I$ that $p$ is covered by $\cover$ and thus by $\cover'$.
In the remaining case, we have $p\in \poly \setminus(\bigcup R_j\cup \bigcup \triangleC(c_i))$, so we get from~\eqref{polyUnion} that $p\in \poly_v$ or $p\in \poly_h$.
Then $p$ is in a trapezoid $T$ of $\mathcal T_h$ or $\mathcal T_v$, so $p$ is in a triangle of $\cover'$ that also covers one of the spikes of $T$.

Suppose then that $\poly'$ has a cover $\cover'$ of size $k'$.
We claim that $s$ polygons in $\cover'$ are used to cover the spikes.
Although the tips of some spikes see each other, we can assume that no polygon in $\cover'$ covers two spike corners, for the following reason:
Two tips that see each other must be tips of spikes added to the same trapezoid in $\mathcal T_v$ or $\mathcal T_h$.
The only convex set that covers two such tips is a line segment.
But that means that if $\piece\in\cover'$ covers two tips, then $\piece$ is a line segment, and then $\cover'\setminus\{\piece\}$ must cover all of $\poly'$.
Hence, we can assume that each tip is covered by a unique polygon in $\cover'$.
We conclude that $s$ polygons are used to cover the spikes.

By construction, each tip sees nothing outside $\poly_v\cup\poly_h$.
Therefore, the remaining $k$ polygons in $\cover'$ must cover $\poly\setminus (\poly_v\cup \poly_h)$.
Since no marked corner is contained in $\poly_v\cup \poly_h$ and $\poly_v\cup \poly_h$ is interior-disjoint from each marked rectangle, we have that $k$ polygons cover $\mathcal C$ and $\mathcal R$.
Hence, $\mathcal I$ has a cover.
\end{proof}

\section{Reduction from {\rangeetrinv}\ to \mrcc}\label{sec:etrinv-to-mrcc}

Let $\II_1$ be an instance of \rangeetrinv\ with an \etrinv\ formula $\Phi$ of $n$ variables $X\mydef\{x_1,\ldots,x_{n}\}$.
We show that there exists an instance $\II_2\mydef\langle\, \poly, \mathcal C, \mathcal R \,\rangle$ of \mrcc\, which can be computed in polynomial time such that $\Phi$ has a solution if and only if $\II_2$ has a cover.
Recall that a cover $\cover$ of $\II_2$ is a set of $|\mathcal C|$ convex polygons contained in $\poly$ that together cover the marked corners $\mathcal C$ and the marked rectangles in $\mathcal R$.
A high-level sketch of the polygon $\poly$ is shown in Figure~\ref{fig:sketch-all}.

\begin{figure}
\centering
\includegraphics[page=18]{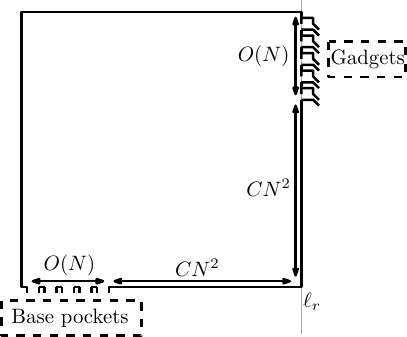}
\caption{%
	A high-level sketch of our construction. The polygon \(\poly\) will have
	dimensions quadratic in \(N\), which is proportional to the size of the \etrinv{} instance \(\Phi\).
}
\label{fig:sketch-all}
\end{figure}

\subsection{Basic tools}\label{sec:basic}

Before describing the actual geometry of the instance $\II_2$, we explain the basic tools that underlie our construction.

\subsubsection{Invisibility property}
We make sure that our constructed instance $\II_2$ satisfies the following \emph{invisibility property}:
No two marked corners can see each other.

For a marked corner $c\in\mathcal C$, let $e_1$ and $e_2$ be the two edges of $\poly$ incident at $c$.
Except for one marked corner in each $\leq$-inversion gadget, it holds that $c$ is the lower endpoint of both $e_1$ and $e_2$.
The invisibility thus follows for all other corners.
Additionally, it holds without exception that $c$ is either the left endpoint of both $e_1$ and $e_2$ or the right endpoint of both.

Note that since a cover $\cover$ for $\II_2$ consists of $|\mathcal C|$ pieces and no piece can cover more than one marked corner by the invisibility property, it follows that each piece $\piece\in\cover$ contains exactly one marked corner $c\in\mathcal C$.

\subsubsection{Critical segments and the bi-cover property}
In the instance $\II_2$ that we construct, a marked rectangle $R\in\mathcal R$ often contains a special line segment $s$ which we call a \emph{critical segment}; see Figure~\ref{fig:criticalSeg}.
The following property turns out to be crucial, and we call it the \emph{bi-cover property}:
For each critical segment $s$, there are at least two and at most three marked corners that can see $s$.
There are three if and only if $s$ is not horizontal, and in that case, one of these three corners $h$ is a \emph{helper corner} (to be described later), and $h$ is incident to an edge contained in the extension of $s$.
In particular, $s$ is on the boundary of the visibility polygon of $h$.
Otherwise, if $s$ is horizontal, then $s$ is a variable segment, to be described in more detail below.

\begin{figure}
\centering
\includegraphics[page=13]{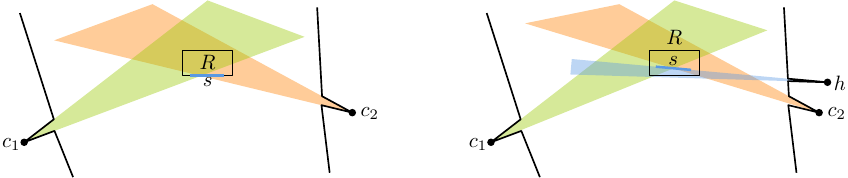}
\caption{%
Left: A horizontal critical segment $s$ contained in the bottom edge of its marked rectangle $R$.
Two pieces cover $R$ and the responsible lever corners $c_1$ and $c_2$.
Right: A critical segment that is not horizontal.
We use a helper corner $h$ with an incident edge contained in the extension of $s$.
The three pieces cover $R$ and the responsible lever corners $c_1,c_2$ and the helper corner $h$.  
}
\label{fig:criticalSeg}
\end{figure}

A critical segment will always be contained in a marked rectangle, but we will also make some marked rectangles that do not contain critical segments.
Often these are introduced to satisfy the bar intersection promise of the \mrcc\ instance, and they will be defined as the intersection of a vertical and a horizontal bar of two other marked rectangles.
Another example is the marked rectangle $R_{\Gamma'}$ in each addition gadget, to be explained in Section~\ref{sec:GEadditionGadget}.

Consider a critical segment $s$ and the marked rectangle $R$ containing $s$.
Let $c_1,c_2\in\mathcal C$ be the marked corners that see $s$ and are not helper corners.
These are called \emph{lever corners}, and we say that $c_1$ and $c_2$ are \emph{responsible} for $s$ and for the marked rectangle $R$.
Since a cover for $\II_2$ must cover all of the rectangle $R$, it follows from the invisibility property and the bi-cover property that $c_1$ and $c_2$ must be covered by two pieces $\piece_1$ and $\piece_2$ that also together cover $s$.
All the marked corners that we make will be either helper corners or lever corners.
A piece that covers a lever corner is called a \emph{lever piece} and otherwise it is a \emph{helper piece}.

\subsubsection{Helper corners}
As mentioned above, we will make two types of marked corners, namely lever corners and helper corners.
Furthermore, we will make two kinds of helper corners.
In general, we make a helper corner when a marked rectangle or a part of a marked rectangle cannot be covered by lever pieces.
One situation is when we introduce a marked rectangle with the sole purpose of satisfying the bar intersection promise described in Section~\ref{sec:mrcc}; see Figure~\ref{fig:helpers} for an example.

The other situation is when we have a marked rectangle $R$ containing a critical segment $s$ which is not horizontal; see Figure~\ref{fig:criticalSeg} (right).
In this case, the two pieces $\piece_1$ and $\piece_2$ covering the lever corners $c_1$ and $c_2$ responsible for $s$ can cover the part of $R$ above or below $s$, but they cannot cover everything on the other side.
We then make a helper corner $h$ so that one of the edges incident at $h$ is contained in the extension of $s$.
Thus, $s$ is on the boundary of the visibility polygon of $h$, and a piece covering $h$ can cover the part of $R$ on the other side of $s$.

If a critical segment $s$ is horizontal, we make the marked rectangle $R$ so that $s$ is contained in the bottom or top edge of $R$, as seen in Figure~\ref{fig:criticalSeg} (left).
Then the pieces covering the responsible corners $c_1$ and $c_2$ will be able to cover all of $R$ without the use of a helper corner.

\begin{figure}
\centering
\includegraphics[page=25]{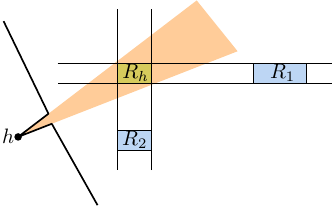}
\caption{%
In order to satisfy the bar intersection promise, we make a marked rectangle $R_h$ which is the intersection of the horizontal bar of $R_1$ and the vertical bar of $R_2$.
We then make a helper corner $h$ so that a piece covering $h$ can also cover $R_h$.
}
\label{fig:helpers}
\end{figure}

\subsubsection{Representing variables}
Each variable $x\in X$ is represented by a collection of \emph{variable segments}, which are horizontal critical segments.
Consider one variable segment $s\mydef ab$ where $a$ is to the left of $b$, and assume that $s$ represents $x$.
The segment $s$ can be oriented either to the left or to the right.
Each point $p$ on $s$ represents a value of $x$ in the interval $[\frac 12,2]$, which we denote as $\xx sp$.
If $s$ is right-oriented, then we define $\xx sp\mydef 1/2 +\frac{3\|ap\|}{2\|ab\|}$, and otherwise we define $\xx sp\mydef 1/2+\frac{3\|bp\|}{2\|ab\|}$.
In particular, $\xx s\cdot$ is a linear map from $s$ to $[\frac 12,2]$.

Consider a variable segment $s$ and a lever corner $c$ responsible for $s$.
In a cover $\cover$ for $\II_2$, the piece $\piece$ covering $c$ specifies a value at $s$, defined as follows.
Each lever corner comes with a \emph{pivot} which is also a corner of $\poly$.
Let $d$ be the pivot of $c$.
It is then always the case that $d$ has another $y$-coordinate than $s$.
For a point $q\in s$, consider the line $\overleftrightarrow{qd}$ through $q$ and the pivot $d$, which is never horizontal by the previous remark, so the line has a well-defined left and right side.
Furthermore, it will be the case that whether $c$ is to the right or left of $\overleftrightarrow{qd}$ does not depend on the particular point $q$.
We say that $c$ is the \emph{right-responsible} or \emph{left-responsible} for $s$, depending on this side.
For instance, in Figure~\ref{fig:ineqChain}, $c_c$ is right-responsible and $c_b$ is left-responsible for the segment $s_3$.

If $c$ is a right-responsible, then we define $\ppCover{s}{c}{\cover}$ as the leftmost point in $\piece\cap s$ (or the right endpoint of $s$ if $\piece\cap s=\emptyset$).
Otherwise, we define $\ppCover{s}{c}{\cover}$ as the rightmost point in $\piece\cap s$ (or the left endpoint of $s$ if $\piece\cap s=\emptyset$).
We now define the value represented by $\cover$ at $s$ with respect to $c$ as $\valCover sc{\cover}\mydef \xx{s}{\ppCover{s}{c}{\cover}}$.
We will usually use the simplified notation $\pp sc$ and $\val sc$ if $\cover$ is clear from the context.

Recall that for each critical segment $s$, there are two responsible lever corners $c_1,c_2$, so that in a cover for $\II_2$, the pieces covering $c_1$ and $c_2$ must together cover $s$.
It will always be the case in our construction that one is left-responsible and the other is right-responsible.
We then have the following observation.

\begin{observation}\label{obs:ineq1}
Consider a variable segment $s\mydef ab$, where $a$ is the left endpoint, and let $c_1$ be the left-responsible for $s$ and $c_2$ the right-responsible.
Consider any cover $\cover$.
Since the pieces covering $c_1$ and $c_2$ also cover $s$, we have $\pp s{c_2}\in a\pp s{c_1}$.
It follows that if $s$ is right-oriented, then $\val s{c_1}\geq \val s{c_2}$.
Otherwise, $\val s{c_1}\leq \val s{c_2}$.
\end{observation}

We are now ready to state the theorem that expresses that our reduction works as desired.

\begin{theorem}\label{thm:mainthm}
Suppose that there is a solution to $\Phi$.
Then our constructed instance $\II_2\mydef\langle\, \poly, \mathcal C, \mathcal R \,\rangle$ has a cover, and any cover $\cover$ of $\II_2$ \emph{specifies} a solution to $\Phi$, as described by the following two properties.
\begin{itemize}
\item Each variable $x\in X$ is specified \emph{consistently} by $\cover$: For each variable segment $s$ representing $x$, consider the corners $c_1,c_2\in\mathcal C$ that are responsible for $s$.
Then $\xx s{c_1}=\xx s{c_2}$, and this value is the same for all segments $s$ representing $x$.

\item The cover $\cover$ is \emph{feasible}, i.e., the values of $X$ thus specified is a solution to~$\Phi$.
\end{itemize}
On the other hand, if there is no solution to $\Phi$, then there is no cover for $\II_2$.
The instance $\II_2$ can be computed in time polynomial in the input size $|\II_1|$.
It thus follows that \mrcc\ is $\ER$-complete.
\end{theorem}

\subsubsection{Lever mechanism}
The basic mechanism of our construction is similar to that of a lever (as known from mechanics) and works as follows; see Figure~\ref{fig:lever}.
As mentioned earlier, a marked corner will either be a helper corner or a \emph{lever corner}, and each lever corner has a \emph{pivot} $d$, which is also a corner of $\poly$ but not a marked corner.
The two corners $c$ and $d$ can always see each other.
There is also a set of marked rectangles $\mathcal S\mydef\{R_1,\ldots,R_l\}$, $l\geq 2$, for which $c$ is one of the responsible corners.
In particular, the rectangles $\mathcal S$ can be seen from $c$.
One or more of the rectangles $\mathcal S_t\mydef\{s_1,\ldots,s_m\}$, $m<l$, are above $d$ and one or more $\mathcal S_b\mydef\{s_{m+1},\ldots,s_l\}$ are below $d$.
Recall that some marked rectangles contain variable segments.
The variable segments contained in the rectangles $\mathcal S$ all represent the same variable $x\in X$, and the variable segments in $\mathcal S_t$ have the same orientation, which is the opposite of that of the variable segments in $\mathcal S_l$.

\begin{figure}
\centering
\includegraphics[page=3]{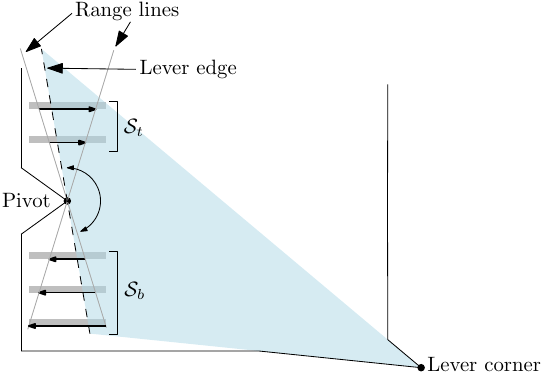}
\caption{A lever piece partially covering the associated marked rectangles, each of which contains a variable segment.
Two overlapping and maximal lever pieces are shown, and the lever edge of one of them is shown as a dashed segment.}
\label{fig:lever}
\end{figure}

A piece of a cover for $\II_2$ that covers a lever corner $c$ is called a \emph{lever piece}.
A lever piece $\piece$ can cover a part of the rectangles $\mathcal S$, but the pivot $d$ is preventing $\piece$ from covering all the rectangles entirely.
Suppose that we want the piece $\piece$ to have non-empty intersection with all rectangles in $\mathcal S$ and cover as much as possible of each one.
Such a piece must have an edge $e$ that contains $d$, as we could otherwise cover more of the top rectangles $\mathcal S_t$ or the bottom rectangles $\mathcal S_b$.
We say that $e$ is a \emph{lever edge}, and $e$ has the property that it intersects all the rectangles $\mathcal S$.
Furthermore, the lever piece $\piece$ contains the part of each rectangle $R$ on the same side of $e$, either to the left (if $c$ is the left-responsible for the rectangles $\mathcal S$) or to the right (if $c$ is the right-responsible).
We can now imagine that the edge $e$ rotates around the pivot $d$.
Doing so, the piece $\piece$ will cover more of the top rectangles $\mathcal S_t$ and less of the bottom rectangles $\mathcal S_b$, or vice versa.

Consider the subset $\mathcal S'\subseteq \mathcal S$ that contain critical segments (the reader can conveniently imagine the case $\mathcal S'=\mathcal S$; the only exception will be in the addition gadgets).
Then there are two \emph{range lines}, both of which pass through $d$, and each critical segment in $\mathcal S'$ has an endpoint on each range line.
When the lever edge of $\piece$ coincides with a range line, then the critical segments in one of the sets $\mathcal S_b$ or $\mathcal S_t$ are completely covered, while the segments in the other set are not covered at all.

\begin{figure}
\centering
\includegraphics[page=17]{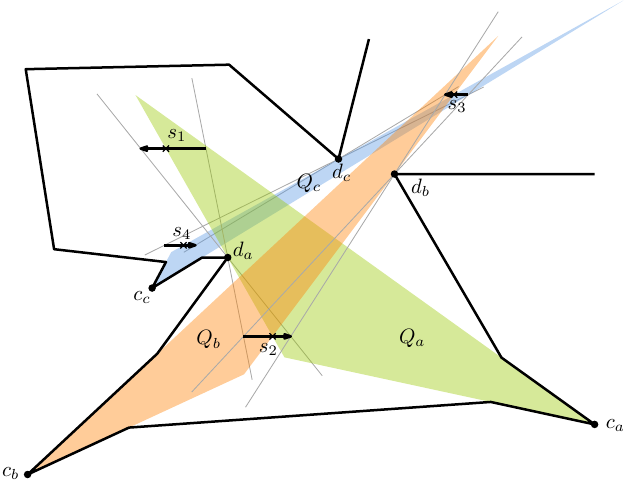}
\caption{Three lever corners and lever pieces that together cover the variable segments $s_2$ and $s_3$.}
\label{fig:ineqChain}
\end{figure}

\subsubsection{Chains of inequalities}

The following observation follows immediately from the description above and is one of our main tools for copying values from one variable segment to another.

\begin{observation}\label{obs:ineq2}
Let $R_t\in\mathcal S_t$ and $R_b\in\mathcal S_b$ be two marked rectangles and assume that they contain variable segments $s_t$ and $s_b$, respectively. 
Suppose that the lever corner $c$ is the right-responsible for both $s_t$ and $s_b$.
If $s_t$ is right-oriented and $s_b$ is left-oriented, then $\val {s_t}c \geq \val {s_b}c$.
If $s_t$ is left-oriented and $s_b$ is right-oriented, then $\val {s_t}c \leq \val {s_b}c$.
If on the other hand $c$ is the left-responsible, then the inequalities are swapped.
Equalities hold if and only if $\piece$ has a lever edge, and then $\piece$ represents the same value on all variable segments in $\mathcal S$.
\end{observation}

Using chains of lever mechanisms, we can now make chains of inequalities of the values that the lever pieces represent of a variable $x\in X$, using Observations~\ref{obs:ineq1} and~\ref{obs:ineq2} alternatingly.
See for instance Figure~\ref{fig:ineqChain}.
The observations give $\val {s_1}{c_a} \leq \val {s_2}{c_a}\leq \val {s_2}{c_b} \leq \val {s_3}{c_b}\leq \val{s_3}{c_c} \leq \val{s_4}{c_c}$.
The figure shows the case that $Q_a,Q_b,Q_c$ all have lever edges (containing the pivots $d_a,d_b,d_c$, respectively) and cover $s_2$ and $s_3$ with no overlap, and then we have equalities $\val {s_1}{c_a} = \val {s_2}{c_a}= \val {s_2}{c_b} = \val {s_3}{c_b} = \val{s_3}{c_c} = \val{s_4}{c_c}$.
Otherwise, one or more of the inequalities will be strict.

\begin{figure}
\centering
\includegraphics[page=30]{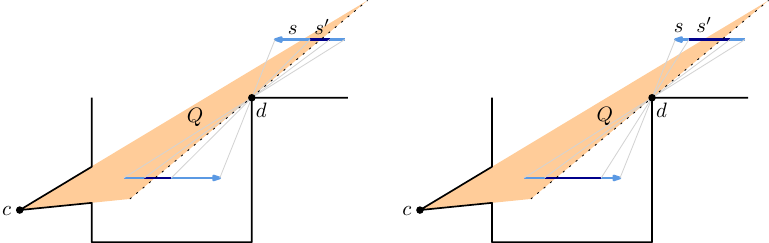}
\caption{The (light and dark) blue arrows show the full variable segments representing a variable $x$ and the dark blue parts are the parts corresponding to the restricted range $I(x)$.
To the left, the dotted lever edge has enough freedom to rotate over the entire restricted range, since the corner $c$ sees all points on the restricted ranges of the variable segments.
To the right, $c$ does not see all of the upper restricted range $s'$, so not all values can be represented.
}
\label{fig:flexibility}
\end{figure}

\subsubsection{Infinitesimal range}
Recall that each variable $x\in X$ comes with an interval $I(x)\subset [\frac 12,2]$ of size $|I(x)|\leq \delta$, for some small $\delta$, such that all solutions to $\Phi$ are guaranteed to be contained in these intervals.
In this reduction, we use $\delta\mydef n^{-7}$.
Consider a variable segment $s$ representing $x$.
Let $s'$ be the part of $s$ such that the points of $s'$ correspond to the range $I(x)$.
We then call $s'$ the \emph{restricted range} of $s$, and since $\delta\mydef n^{-7}$, we have that $\frac{\|s'\|}{\|s\|}=O(n^{-7})$.
Since we can assume $n$ to be arbitrarily large, we can usually think of $s'$ as being just a single point on $s$.
In our construction, we will always show the full segment $s$ in the figures, and we will mark the position of the restricted range $s'$ as a single point on $s$.
We just need to ensure that the construction is generic in the sense that it works for every possible placement of the restricted range $s'$ on the full range $s$.
In particular, when we construct a lever corner $c\in\mathcal C$, it may be needed to move $c$ and the edges of $\poly$ incident at $c$ depending on where the restricted range $s'$ is located at $s$.

We need the lever corner $c$ to be placed so that it can see all points on the restricted range $s'$ of $s$.
When that is the case for all segments $s$ that $c$ is responsible for, then a lever edge of a piece covering $c$ will have the freedom to rotate over all of $s'$, thus representing all values in the required range $I(x)$.
See Figure~\ref{fig:flexibility} for examples with enough freedom and too little freedom.

In many parts of our construction, it is obvious that $c$ can see all of the restricted ranges of the variable segments.
The more delicate case happens for the propagation corners, to be described in more detail in the next section.
These are responsible for variable segments in the base pockets, but also for critical segments in the corridors.
Here, it is not clear that the propagation corners can see enough of the critical segments in the corridors, but as we will see in Section~\ref{sec:propSpikes}, it works for our choice of $\delta\mydef n^{-7}$.

We choose the marked rectangles to be similarly tiny, as they only need to contain the restricted ranges of the critical segments.
Hence, the vertical and horizontal bars of the rectangles can be thought of as vertical or horizontal line segments contained in $\poly$.

\subsection{High-level description of $\poly$}

\begin{table}[p]
\centering
\begin{tabular}{|l|l|}
\hline
Name & Description/value   \\
\hline
$\II_1\mydef \langle\, \Phi,\delta,(I(x_1),\ldots,I(x_n))\, \rangle$ & instance of \rangeetrinv\ that we reduce from \\
$\II_2\mydef\langle\, \poly, \mathcal C, \mathcal R \,\rangle$ & instance of \mrcc\ that we reduce to \\
$X\mydef \{x_1,\ldots,x_{n}\}$ & set of variables of $\Phi$ \\
$n$ & number of variables in $\Phi$ \\
$\Phiineq$ & formula equivalent to $\Phi$ using only inequalities \\
$\Xineq\mydef \{x_1,\ldots,x_{3n}\}$ & variables of $\Phiineq$ \\ 
$\delta$ & $n^{-7}$, maximum range of each variable $x_i$ \\
$\poly \mydef \poly(\Phi)$ & final polygon to be constructed from $\Phi$ \\
$C$ & a sufficiently large constant \\
$N$ & maximum index of a variable segment $s_i$ in the base pockets, $N=O(n^3)$ \\
$s_i \mydef a_ib_i$ & variable segment in a base pocket, $i\in\{1,\ldots,N\}$ \\ 
$c_i^s$ & propagation corner of the segment $s_i\mydef a_ib_i$ \\
$o$ & $c_1^s$, relative origin of the base pockets \\
$\tilt$ & $1/32$, every variable segment $s_i$ has a length in $[\frac 32(1-\tilt),\frac 32]$ \\
corridor  & connection between the main part of $\poly$ and a gadget \\
$d_0 d^0, d_1 d^1$ & corridor doors, $d^0 \mydef d_0 + (0,\frac{4.5}{CN^2})$ and $d^1 \mydef d_1 + (0,\frac{2.25}{CN^2})$ \\
$\ell_{\proj o}$ & vertical line through the point $\frac{d_0+d_1}{2}$ \\
$\proj o$ & intersection of the ray $\overrightarrow{od_0}$ and $\ell_{\proj o}$ \\
$\bar o$ & point on $\overrightarrow{\proj o d_1}$ with $x(\bar o) =x(d_1)+1$, relative origin of a particular gadget \\
$r_i,r_j,r_l$ & variable segments within gadgets \\ 
$\bar a_\sigma,\bar b_\sigma$, $\sigma\in\{i,j,l\}$ & left and right endpoint of $r_\sigma$, $\sigma\in\{i,j,l\}$ \\
$c_\sigma^r$ & right lever corner responsible for $r_\sigma\mydef \bar a_\sigma\bar b_\sigma$ \\
$\nu$, $\rho$ , $\eps$ & $\nu \mydef \frac{13.5}{CN^2}$, 
$\rho \mydef\frac{\nu}{9} = \frac{1.5}{CN^2}$, 
$\eps \mydef \frac{\rho}{48} = \frac{1}{32CN^2}$     \\
$\proj V$ & $\proj{o}+[-N\nu,N\nu]\times [-N\nu,N\nu]$, square in a corridor \\
slab $S(q,v,r)$ & region of all points with distance at most $r$ to the line through $q$\\
&with direction $v$ \\
center of slab& line in the middle of a slab\\
$L$-slabs, $R$-slabs & uncertainty regions for visibility rays, see Section~\ref{sec:describingVisibilities} \\
\hline
\end{tabular}
\caption{Parameters, variables, and certain distances that are frequently used are summarized in this table for easy access. Some descriptions are much simplified.}
\label{tab:Glossary}
\end{table}

We are now ready to describe the actual design of the polygon $\poly$ in more detail; see Figure~\ref{fig:sketch-all} for a sketch of the entire polygon and Table~\ref{tab:Glossary} for some of the terms and symbols that are often used.
In the bottom of $\poly$, there are pockets with lever corners representing all the variables; see Figure~\ref{fig:basepocket}.
Each pocket is a polygon in $\poly$ bounded by a chain of $\partial\poly$ and from above by one horizontal chord of $\poly$.
A pocket can contain a large number of variable segments $\mathcal S$, all representing the same variable $x\in X$.
In the lower right corner of the pocket, there is a lever corner, which we call a \emph{base corner}, and which is the right-responsible lever corner for all the segments $\mathcal S$.
A piece covering a base corner is called a \emph{base piece}.
For each variable segment, the left-responsible lever corner is called a \emph{propagation corner}.
The pivot of a propagation corner is very far away approximately in direction $(1,1)$, namely at the door of a corridor that leads into a gadget.
A piece covering a propagation corner is called a \emph{propagation piece}, and in any cover, the propagation pieces will be very skinny and long, sticking out of the base pockets and into the corridors, where each of them is responsible for covering another critical segment.

\begin{figure}
\centering
\includegraphics[page=5]{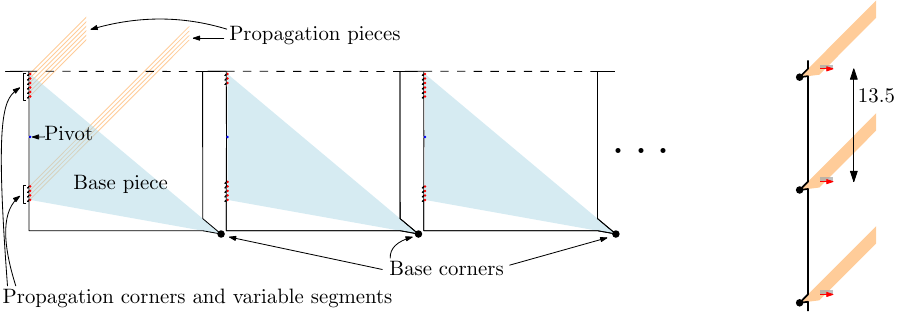}
\caption{Left: A sketch of three consecutive base pockets with various numbers of variable segments.
The variable segments are shown as red dots.
Right: A closeup of three consecutive variable segments and the marked rectangles containing them.
}
\label{fig:basepocket}
\end{figure}

The bottom part of $\poly$ consists of a collection of $3n$ base pockets.
In order from left to right, we denote the pockets as $\pocket_1,\ldots,\pocket_{3n}$.
The pockets $\pocket_1,\ldots,\pocket_{n}$ represent the variables $x_1,\ldots,x_{n}$, respectively, as do the pockets $\pocket_{n+1},\ldots,\pocket_{2n}$, and $\pocket_{2n+1},\ldots,\pocket_{3n}$.
At the right side of $\poly$, there are some \emph{corridors} attached, each of which leads into a \emph{gadget}.
The \emph{doors} to the corridors are line segments contained in a vertical line $\ell_r$ (in fact, we will move the doors slightly away from $\ell_r$ in order to satisfy the trapezoid generality promise, but that does not make any conceptual difference).
The gadgets also contain marked rectangles and corners, and they are used to impose dependencies between the pieces covering variable segments in the base pockets.
The corridors are used to make dependencies between the base pockets and the gadgets.
Each gadget corresponds to a constraint of one of the types $x+y\geq z$, $x+y\leq z$, $x\cdot y\geq 1$, $x\cdot y\leq 1$, and $x\leq y$.
The first four types of constraints are used to encode the dependencies between the variables in $X$ as specified by $\Phi$, whereas the latter constraint is used to make sure that the three base pockets representing a variable $x\in X$ specify the value of $x$ consistently.

The reason that we need three pockets $\pocket_i,\pocket_{i+n},\pocket_{i+2n}$ representing each variable $x_i$ is that in the addition gadgets, such as the one encoding the restriction $x_i+x_j\geq x_l$, we need the base pockets representing $x_l,x_j,x_i$, that we connect to the gadget, to appear in the specific order $\pocket_{i'},\pocket_{j'},\pocket_{l'}$ from left to right, respectively.
This is obtained by choosing $i'\mydef i$, $j'\mydef j+n$, and $l'\mydef l+2n$.

\subsection{Detailed construction of the bottom wall}\label{sec:bottomWall}
In this section we present the construction of the bottom wall of $\poly$.
Recall that the bottom wall forms $3n$ base pockets $\pocket_1,\ldots,\pocket_{3n}$, such that a variable $x_i$ is represented by the pockets $\pocket_i,\pocket_{n+i},\pocket_{2n+i}$.
Each pocket may contain multiple variable segments.

We define a formula $\Phiineq$ equivalent to the \etrinv-formula $\Phi$ in the following way.
We formally introduce $2n$ new variables $x_{n+1},\ldots,x_{3n}$, such that each pocket $\pocket_i$ represents just one variable $x_i$, for $i\in\{1,\ldots,3n\}$.
We define $\Xineq\mydef\{x_1,\ldots,x_{3n}\}$ as the total set of variables.
For each $i\in\{1,\ldots,n\}$, we introduce the three constraints $x_i\geq x_{n+i}$, $x_{n+i}\geq x_{2n+i}$, and $x_{2n+i}\geq x_i$.
These ensure that $x_i=x_{n+i}=x_{2n+i}$, which corresponds to the pockets $\pocket_i,\pocket_{n+i},\pocket_{2n+i}$ specifying the same value of a variable.
We rewrite every occurrence of $x_i+x_j=x_l$ as two inequalities:
$$x_{i}+x_{n+j}\geq x_{2n+l}\text{\quad and\quad } x_{i}+x_{n+j}\leq x_{2n+l}.$$
This is to ensure that for all addition inequalities, the indices of the appearing variables have a specific order.
Likewise, we rewrite every occurrence of $x_ix_j=1$ as two inequalities:
$$x_ix_{n+j}\geq 1\text{\quad and\quad }x_ix_{n+j}\leq 1.$$

To sum up, we now have a formula $\Phiineq$ equivalent to $\Phi$, and $\Phiineq$ is a conjunction of inequalities of the forms
$$
x_i+x_j\geq x_l,\quad x_i+x_j\leq x_l,\quad x_ix_j\geq 1,\quad x_ix_j\leq 1, \quad x_i\geq x_j,
$$
where $i,j,l\in\{1,\ldots,3n\}$.
The last inequality is only used to ensure $x_i=x_{n+i}=x_{2n+i}$ for each $i\in\{1,\ldots,n\}$, as described before.
For each addition inequality, we have $i<j<l$, and for each multiplication inequality we have $i<j$.
We have two different types of inequalities of the form $x_i\geq x_j$, namely those where $i<j$ and those where $j<i$.
We can assume that there are at most $O(n^3)$ inequalities in $\Phiineq$, as otherwise some inequalities would be identical.

Consider a variable $x_i$, $i\in\{1,\ldots,3n\}$, and the appearances of the variable in the inequalities of $\Phiineq$, rewritten as described above.
We associate each appearance of $x_i$ to a distinct variable segment in the pocket $\pocket_i$.
For each inequality $\xi$, there will be a gadget $G_\xi$ that enforces the base pockets to obey the inequality $\xi$.
The variable segments associated to the appearances of the variables in $\xi$ will be connected to $G_\xi$ via a corridor.
Whether the segments associated to $\xi$ should be right- or left-oriented will be clear when the individual gadgets are described in detail.

Recall that a variable segment has endpoints on the two range lines, so the length of the segment depends on its $y$-coordinate.
We want all variable segments in the base pockets to have a length in the range $[\frac 32(1-\tilt),\frac 32]$, for the small constant $\tilt\mydef \frac 1{32}$.
A pocket has to be dimensioned correctly in order to accommodate enough segments of this length.
To keep the construction simpler, we make all the pockets equally large, and this size is chosen so that all pockets can accommodate their segments.
The construction of the bottom wall will therefore depend on the pocket that receives the largest number $k$ of variable segments with the same orientation.


We place the right-oriented segments above the pivot and the left-oriented below the pivot.
In each pocket, we place the variable segments with the same orientation equidistantly with vertical distance $13.5$, with the right-oriented segments in the top and the left-oriented in the bottom; see Figure~\ref{fig:basepocket}.
We therefore choose the height of the pockets to be $13.5k/\tilt$.
This ensures that even in the pocket with most segments of the same orientation, we can place the right-oriented segments within distance $13.5k$ from the top and the left-oriented segments within distance $13.5k$ from the bottom, and then the length of the shortest will be at least $1-\tilt$ times the longest.
We choose the width of the pockets slightly larger than the height, so that rays with direction $(1,1)$ from the bottom-most variable segment in one pocket are not obstructed by the pocket.
This is important since we are going to place the gadgets so that the propagation pieces will follow approximately the direction $(1,1)$.
Furthermore, the rays from the bottom-most segment in pocket $P_i$ will then be above the rays from the top-most segment in the pocket $P_{i+1}$ to the right.
Note that since $\tilt$ is small, the range lines (i.e., the two lines through the pivot that together contain all endpoints of the segments in one base pocket) will be close to vertical and the variable segments will appear as a vertical stack of equally long segments.

We place the base pockets so that all pivots are on a horizontal line and the distance between two consecutive pivots is a multiple of $13.5$, and we place the variable segments so that the topmost segments in all pockets are on a horizontal line, as are the bottom-most.
We want to enumerate the variable segments according to the slope of the direction to any particular corridor.
We place the segments such that for each segment $s$, the midpoint is vertically above or below the pivot of the base pocket.
It then follows that these midpoints also have $x$-coordinates that are multiples of $13.5$.
For a particular variable segments with midpoint $p=(x(p),y(p))$, we then give a \emph{score} defined as $x(p)-y(p)$.
This ensures that every variable segments gets a distinct score which is a multiple of $13.5$.
Let $w_{\min}$ and $w_{\max}$ be the minimum and maximum scores, respectively, among all segments in all pockets.
We then denote a variable segment with score $w$ as $s_i$ where $i\mydef (w-w_{\min})/13.5+1$.
Let $N\mydef (w_{\max}-w_{\min})/13.5+1$, so that each segment $s_i$ has an index $i\in \{1,\ldots,N\}$.
Many of these indices are unused, because we only placed segments close to the top and bottom of each pocket and because the pockets are wider than they are high.
We have that $N=O(n^3)$, since each variable can participate in at most $O(n^2)$ different gadgets, so there are at most $O(n^2)$ segments in each base pocket, and there are $O(n)$ base pockets.

We define the left and right endpoints of the segment to be $a_i$ and $b_i$, where $a_i$ is to the left.
Consider a variable segment $s_i$ and let the midpoint be $p$.
We then define the associated propagation corner of $s_i$ as $c_i^s\mydef p+(-13/4,-1)$.
We choose the neighbouring corners to be $c_i^s+(1,\frac 18)$ and approximately $c_i^s+(1,1)$ (the precise way to choose this second corner will be described in Section~\ref{sec:propSpikes}).

It will be convenient to define the relative origin of the base pockets to be the first propagation corner $o\mydef c_1^s$.
We then have the following observation by the above description:

\begin{observation}\label{obs:basePocketLines}
For each $i\in\{1,\ldots,N\}$, if the segment $s_i$ exists, there exists $\mu\in [0,\tilt/2]$ such that the following holds:
\begin{itemize}
\item
$c_i^s$ is on the line through $o+(13.5(i-1),0)$ with direction $(1,1)$,

\item
$a_i$ is on the line through $o+(13.5(i-1)+1.5+\mu,0)$ with direction $(1,1)$,

\item
$b_i$ is on the line through $o+(13.5(i-1)+3-\mu,0)$ with direction $(1,1)$.
\end{itemize}
\end{observation}

\subsection{Consistent variable representation in base pockets}\label{sec:consistency}

Before we explain how to make corridors and gadgets, let us assume that we already have them at hand and show how we can finish the construction using them.
Recall that we have rewritten the original \rangeetrinv\ formula $\Phi$ with variables $X\mydef\{x_1,\ldots,x_n\}$ as a formula $\Phiineq$ with variables $\Xineq\mydef\{x_1,\ldots,x_{3n}\}$, so that in all solutions and for each $x_i\in X$, we will have $x_i=x_{i+n}=x_{i+2n}$.
Each variable $x_i\in\Xineq$ is represented by a base pocket $P_i$, and the pockets appear in the order $P_1,\ldots,P_{3n}$ from left to right in the bottom of $\poly$.
In order to ensure that the base triangles in the pockets $P_i,P_{i+n},P_{i+2n}$ represent the same value of $x_i,x_{i+n},x_{i+2n}$, we make use of $\geq$- and $\leq$-gadgets as described in the following and shown schematically in Figure~\ref{fig:basepocketconnection}.
(The gadgets will be described in Sections~\ref{sec:Geinequality} and \ref{sec:Leinequality}.)

In the base pockets, the variable segments above the pivots are right-oriented and the ones below the pivots are left-oriented.
We connect the bottom-most segment of $P_i$ and the topmost in $P_{i+n}$ with a $\geq$-inequality gadget, and connect $P_{i+n}$ and $P_{i+2n}$ in a similar way.
We connect the topmost segment in $P_i$ and the bottommost in $P_{i+2n}$ with a $\leq$-inequality gadget.
We then have a cycle of inequalities that ensure that the base triangles in the three pockets $P_i,P_{i+n},P_{i+2n}$ represent the value of $x_i$ consistently on all variable segments in the three pockets.

When we add the addition and inversion gadgets, the values of the variables $X$ represented by the base triangles must satisfy all the constraints in $\Phi$.
It follows that if there is a cover for $\II_2$, then there exists a solution to $\Phi$.
On the other hand, if there is a solution to $\Phi$, we get a corresponding cover where the pieces represent the given solution.
Thus, Theorem~\ref{thm:mainthm} follows.

\begin{figure}
\centering
\includegraphics[page=24]{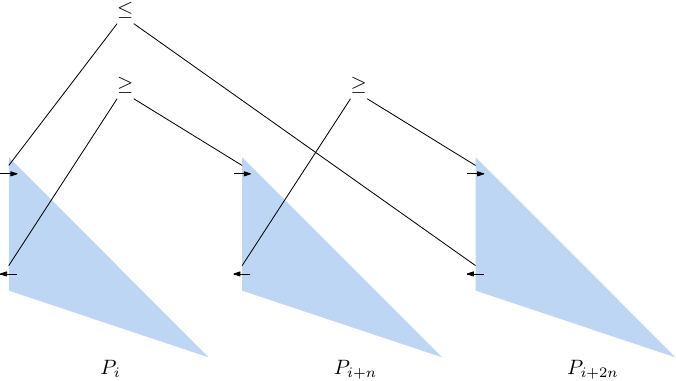}
\caption{Schematic view of how we connect the base pockets representing the same variable $x_i$ using inequality gadgets.
This creates a cycle of inequalities so that the blue base triangles must represent the value of $x_i$ consistently on all variable segments in the pockets.
}
\label{fig:basepocketconnection}
\end{figure}

\subsection{Technique for connecting base pockets and gadgets}
We already saw in Section~\ref{sec:basic} how we can make chains of inequalities using variable segments that are covered pairwise by lever pieces.
When connecting the base pockets and the gadgets, it is advantageous to introduce some critical segments that are not horizontal and hence not variable pieces according to our definition.
In this section, we explain the technique behind using these critical segments.
Such a segment $s$ will, however, represent a variable $x\in X$, but the map from $s$ to $[\frac 12,2]$ will not be linear but some other rational function.
In Section~\ref{sec:copy} we show how to use the technique to make a corridor that can connect multiple variable segments in the base pockets to segments in a gadget.

To capture the dependencies that a critical segment can make between the two lever pieces covering it, we introduce the notion of projections.
Consider a lever corner $c$ and two critical segments $s_0$ and $s_1$ such that $c$ is responsible for both.
Let $d$ be the pivot of $c$.
For $p\in s_0$, we define $\pi(p)\mydef \overleftrightarrow{pd}\cap s_1$, and we say that the function $\pi: s_0\longrightarrow s_1$ thus defined is the \emph{projection} from $s_0$ to $s_1$.
Using generalizations of Observations~\ref{obs:ineq1} and~\ref{obs:ineq2}, we get the following, illustrated in Figure~\ref{fig:nook_and_umbra}.

\begin{figure}
\centering
\includegraphics[page=2]{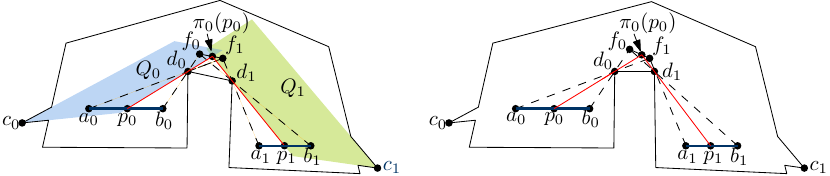}
\caption{%
Left:
If the pieces $\piece_0$ and $\piece_1$ cover all of $f_0f_1$, then the leftmost point of $s_1$ covered by $\piece_1$ is on the segment $p_1 b_1$.
Right:
Lemma~\ref{lemma:copyLemma} says that when $y(d_0)=y(d_1)$, then $\xx {s_0}{p_0} = \xx {s_1}{p_1}$.
}
\label{fig:nook_and_umbra}
\end{figure}

\begin{observation}\label{obs:ineq3}
Consider a critical segment $f\mydef f_0f_1$ and let $c_0$ be the left-responsible for $f$ and $c_1$ the right-responsible.
Let $d_0,d_1$ be the pivots of $c_0,c_1$, respectively, and suppose that $f$ is above both $d_0$ and $d_1$.
Let $s_0\mydef a_0b_0$ and $s_1\mydef a_1b_1$ be variable segments below $d_0$ and $d_1$ for which $c_0$ and $c_1$ are responsible, such that $a_0,a_1$ are the left endpoints and $f_0= \overrightarrow{b_0d_0}\cap \overrightarrow{b_1d_1}$ and $f_1=\overrightarrow{a_0d_0}\cap\overrightarrow{a_1d_1}$.
Let $\pi_0$ be the projection from $s_0$ to $f$ and $\pi_1$ be that from $f$ to $s_1$.
Consider any cover $\cover$, 
and let $p_0\mydef \pp{s_0}{c_0}$ and $p_1\mydef \pi_1(\pi_0(p_0))$.
Then $f\cap \piece_0\subset f_0\pi(p_0)$, and therefore $\pi(p_0)f_1\subset \piece_1$.
We then have $s_1\cap \piece_1\subset p_1b_1$.
In particular, if $s_1$ is right-oriented, then $\val{s_1}{c_1}\geq \xx{s_1}{p_1}$, and otherwise, $\val{s_1}{c_1}\leq \xx{s_1}{p_1}$.
\end{observation}

The principle behind the following special case was also used in~\cite{abrahamsen2018art}; see Figure~\ref{fig:nook_and_umbra} (right) for an illustration, and Figure~\ref{fig:copy} for an example of how we make a helper spike in order to cover the bounding rectangle $R$ of the critical segment $f$.
Although the context in which we state the lemma is different, the proof is analogous to one given in~\cite{abrahamsen2018art}, so we give the lemma without proof.

\begin{lemma}[\cite{abrahamsen2018art}]\label{lemma:copyLemma}
In the setup of Observation~\ref{obs:ineq3}, if $y(d_0)=y(d_1)$, then $\xx{s_0}{p_0}=\xx{s_1}{p_1}$.
In particular, if $s_0$ and $s_1$ both right-oriented, then $\val{s_0}{c_0}\leq\val{s_1}{c_1}$, and if $s_0$ and $s_1$ both left-oriented, then $\val{s_0}{c_0}\geq\val{s_1}{c_1}$.
\end{lemma}

\begin{figure}
\centering
\includegraphics[page=4]{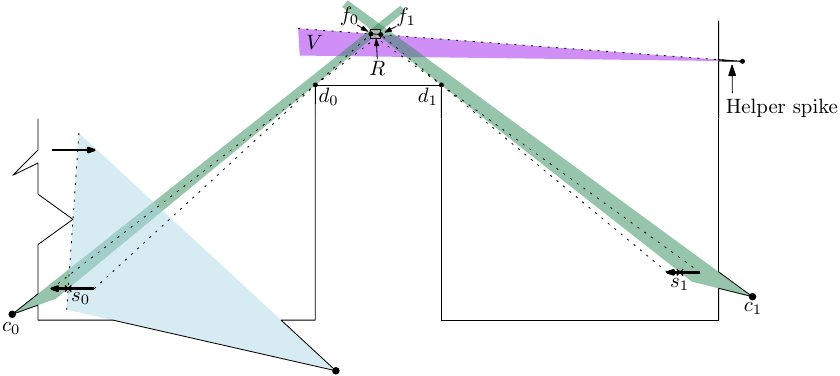}
\caption{%
    The setting of Lemma~\ref{lemma:copyLemma}.
    Let $x_0,x_1$ be the values specified by $U_0$ and $U_1$ on $r_0$ and $r_1$, respectively.
    Then $x_0 \geq  x_1$.
}
\label{fig:copy}
\end{figure}

\subsection{Corridor for copying variable segments into a gadget}\label{sec:copy}

In this section, we describe the construction of a corridor, the purpose of
which is to make inequalities between variable segments from the base pockets and a gadget.
Inside each gadget there are three (or two) variable segments $r_i, r_j , r_l$ (or $r_i, r_j$) corresponding to three (or two) variable segments from the base pockets $s_i,s_j,s_l$ (or $s_i,s_j$).
For each $\sigma\in\{i,j,l\}$, we say that $s_\sigma$ and $r_\sigma$ are \emph{partners}.
The pairs of partners is a matching between the variable segments in the base pockets and those in the gadgets.
%
We require that for each $\sigma \in \{i,j,l\}$ the partners $s_\sigma,
r_\sigma$ have the same orientation since this is required for the way we implement inequalities, as described by Lemma~\ref{lemma:copyLemma}.

We assume $i < j < l$.
We describe here how to construct a corridor that ensures that we have inequalities between the segments
$r_i, r_j, r_l$ and the segments $s_i, s_j, s_l$, respectively.
For any $\sigma\in \{1,\ldots,N\}$, recall that $c_{\sigma}^s$ is the propagation corner responsible for $s_\sigma$, and for each $\sigma\in\{i,j,l\}$, define $c_\sigma^r$ to be the lever corner responsible for $r_\sigma$ to the right, which will be a corner inside the gadget containing $r_\sigma$.
It will follow from Lemma~\ref{lemma:copyLemma} that if $r_\sigma$ and $s_\sigma$ are right-oriented, then $\val {s_\sigma}{c_\sigma^s} \leq \val {r_\sigma}{c_\sigma^r}$.
Otherwise, if they are both left-oriented, we will have $\val {s_\sigma}{c_\sigma^s} \geq \val {r_\sigma}{c_\sigma^r}$.
%


The construction of the corridor relies on similar techniques as described in~\cite{abrahamsen2018art}.
In~\cite{abrahamsen2018art} it was described how to get equality between values represented on each of the pairs of guard segments $(s_i,r_i),(s_j,r_j),(s_l,r_l)$.
In our case, it will suffice to ensure that there is an inequality between the values, using Lemma~\ref{lemma:copyLemma}.
This makes the upper wall of our corridor simpler.

The construction can be generalized for making inequalities between an arbitrary subset of
variable segments, but since the gadgets we use only handle two or three
variable segments at a time, we explain the construction only in the setting of
three segments.
The construction for two segments is analogous but simpler.

The lower wall of the corridor of the gadget is a horizontal edge $d_0d_1$ of length $2$, where $d_0$ is on
the vertical line $\ell_r$ and $d_1$ is to the right of $d_0$.
The upper wall of the corridor is more complicated, and it will be described later.
It has the left endpoint at $d^0\mydef d_0+(0,\frac {4.5}{C N^2})$, and the right endpoint $d^1\mydef d_1+(0,\frac {2.25}{C N^2})$.
We refer to the vertical line segments $d_0d^0$ and $d_1d^1$ as the \emph{doors} of the corridor.

\subsubsection{The principles of the corridor construction}

\begin{figure}
\centering
\includegraphics[page=6]{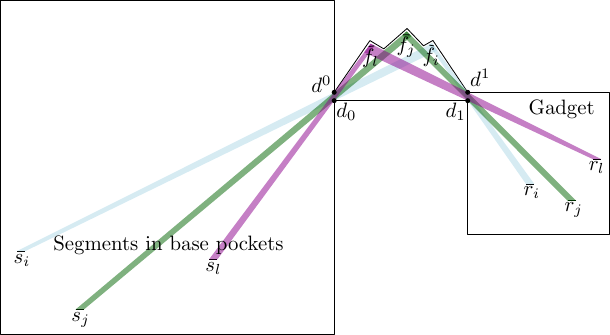}
\caption{%
This figure sketches the corridor construction in a much simplified manner.
The corner $d_0$ serves as a pivot for $s_i,s_j,s_l$ as well as the three critical segments $f_i,f_j,f_l$, and $d_1$ serves as a pivot for $r_i,r_j,r_l$ and $f_i,f_j,f_l$.
The doors $d_0d^0$ and $d_1d^1$ to the corridor are sufficiently small so that the critical segment $f_\sigma$ of the pair $s_{\sigma},r_{\sigma}$ are not seen by the responsible corners of other segments, so they must be covered by pieces covering the corners $c_\sigma^s$ and $c_\sigma^r$.%
}
\label{fig:SimplifiedCopying}
\end{figure}

Figure~\ref{fig:SimplifiedCopying} sketches the construction of the corridor.
The corridor contains three critical segments $f_i,f_j,f_l$, such that the responsible corners of $f_\sigma$ is $c_\sigma^s$ and $c_\sigma^r$.
For each $\sigma\in\{i,j,l\}$, the corner $c_\sigma^s$ is left-responsible for $s_\sigma$ and for $f_\sigma$, with the pivot $d_0$.
Similarly, $c_\sigma^r$ is right-responsible for $r_\sigma$ and $f_\sigma$ and has pivot $d_1$.
We construct the chain of $\poly$ from $d^0$ to $d^1$ bounding the corridor from above so that a reflex corner goes down in between the critical segments which blocks the visibility between the critical segments.
This is to make sure that we satisfy the bar intersection promise described in Section~\ref{sec:mrcc} (that is, by preventing that bars from different marked rectangles overlap, we do not have to introduced more marked rectangles).

The biggest challenge is to ensure the bi-cover invariant of the critical segments $f_i,f_j,f_l$.
In particular, that for each $\sigma\in\{i,j,l\}$, the only marked corners that see $f_\sigma$ is the propagation corner $c_\sigma^s$ responsible for $s_\sigma$ and the lever corner $c_\sigma^r$ responsible for $r_\sigma$.
This is obtained by choosing the doors $d_0d^0$ and $d_1d^1$ correctly.
In particular, we will ensure that the vertical edge of $\poly$ directly below the door $d_0d^0$ blocks visibility from all propagation corners $c_k^s$ for $k>\sigma$, from which the direction to $d_0$ is more vertical than from $c_\sigma^s$.
Likewise, the vertical edge of $\poly$ directly above $d_0d^0$ blocks visibility from all propagation corners $c_k^s$ for $k<\sigma$, from which the direction is more horizontal than from $c_\sigma^s$.
The door $d_1d^1$ in the other side of the pocket similarly ensures that only the intended corner $c_\sigma^r$ can see $f_\sigma$.
We can therefore establish the bi-cover invariant and use Lemma~\ref{lemma:copyLemma}.

The main idea to achieve the above property is to make each door $d_i d^i$, $i\in\{0,1\}$, of the corridor sufficiently small.
However, we cannot place the point $d^0$ arbitrarily close to $d_0$, since then $c_\sigma^s$ will not see enough of the critical segment $f_\sigma$, which is needed. 
By placing the corridor sufficiently far away from the segments of the base pockets, we obtain that the visibility lines from the propagation corners $c_\sigma^s$ through the left door endpoints $d_0,d^0$ are almost parallel to each other and can be described by a simple pattern.
The same holds for the lines from the lever corners $c_\sigma^r$ in the gadgets through the right door endpoints $d_1,d^1$.
The pattern enables us to construct the corridor with the desired properties.

In the following, we introduce objects that make it possible to describe the upper corridor wall and prove that the construction works as intended.

\subsubsection{Describing visibilities in the corridor}\label{sec:describingVisibilities}

%
Recall that $o\mydef c_1^s$ is the propagation corner of the first variable segment $s_1$.
We define the point $\proj o$ as the point on the ray $\overrightarrow{od_0}$ such that $x(\proj o)=x(d_0)+1$.
We consider the rays from the propagation corners in the base pockets through one of the points $d_0$ and $d^0$.
They all pass through a small area around $\proj o$ and have directions close to $(1,1)$.
Similarly, the rays from the gadget (attached to the right of the corridor) through the points $d_1$ and $d^1$ pass through the same region around $\proj o$, and they have directions close to $(-1,1)$.
In this section we introduce some thin slabs with directions $(1,1)$ and $(-1,1)$, and these will contain the rays close to $\proj o$.
The slabs will be used to reason about which lever corners can see which critical segments in the corridor.
%
Thus, the slabs are introduced in order to handle the ``uncertainty'' and irregularity of the rays.

The analysis resembles that from~\cite{abrahamsen2018art}, but is slightly more complicated, since our variable segments have different lengths and are placed at different heights contrary to the guard segments in~\cite{abrahamsen2018art} that are equally long and placed equidistantly on a line.
Furthermore, we need to consider the rays from the propagation corners and variable segment endpoints, whereas in~\cite{abrahamsen2018art} only the guard segment endpoints were considered.
We therefore develop the more general technical Lemma~\ref{lemma:proj} below to deal with this complexity.
The lemma is illustrated in Figure~\ref{fig:proj}.

Given a point $q$ and a vector $v\neq (0,0)$, the \emph{slab} $S(q,v,r)$ consists of all
points at distance at most $r$ from the line through $q$ parallel to $v$.
The \emph{center} of the slab $S(q,v,r)$ is the line through $q$ parallel to $v$.
Let us define vectors $\alpha\mydef (1,1)$ and $\beta\mydef (-1,1)$.

\begin{figure}
	\centering
	\includegraphics[page=19]{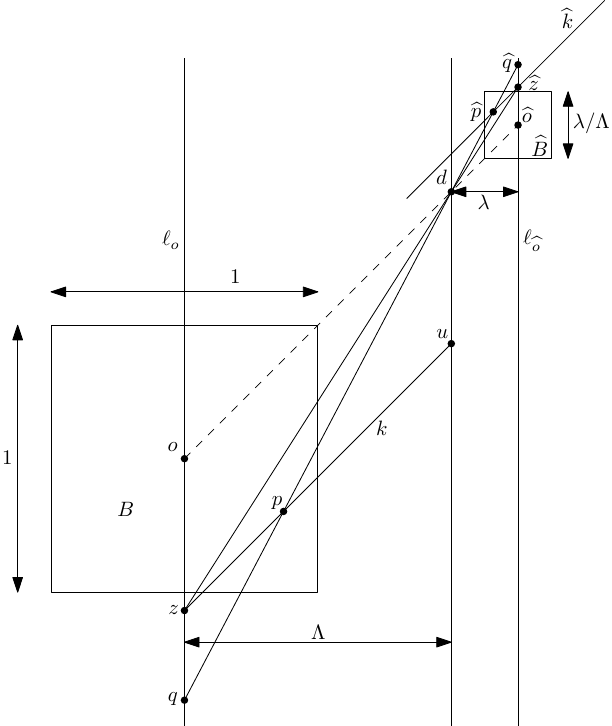}
	\caption{The setting of Lemma~\ref{lemma:proj}.}
	\label{fig:proj}
\end{figure}

\begin{lemma}\label{lemma:proj}
Let $B$ be the axis-parallel unit square with center $o\mydef (0,0)$.
Let $\Lambda>0$ and $\lambda>0$, and let $d$ be a point of the form $\Lambda\alpha+(0,h)$ where $|h|\leq 1$.
Let $\ell_{\proj o}$ be the vertical line with $x$-coordinate $\Lambda+\lambda$.
Let $\proj o$ be the intersection of $\overrightarrow{od}$ and $\ell_{\proj o}$ and let $\proj B$ be the axis-parallel square with center $\proj o$ and edge length $\lambda/\Lambda$.
For a point $p$ to the left of $d$, define $\proj p\mydef \proj o-p\lambda/\Lambda$.
For some $\eps=O(\lambda/\Lambda^2)$, it then holds that for every point $p\in B$, we have $\overrightarrow{pd}\cap \proj B\subset S(\proj p,\alpha,\eps)$.
\end{lemma}

\begin{proof}
Let $\ell_o$ be the vertical line containing $o$.
For a point $p\in B$, note first that $\proj p$ is on the line $\overleftrightarrow{pd}$.
Let $q$ be the intersection point of $\overleftrightarrow{pd}$ and $\ell_o$.
Let $k$ and $\proj k$ be the lines through $p$ and $\proj p$, respectively, with direction $\alpha$, and note that $\proj k$ is the center of any slab $S(\proj p,\alpha,\eps)$.
Let $z$ be the intersection point of $\ell_o$ and $k$. 
We first bound the distance $\|\widehat z\proj q\|$, which is the vertical deviation between $\proj k$ and $\overrightarrow{pd}$ at $\ell_{\proj o}$.
Since the triangles $\proj p \proj z\proj q$ and $pzq$ are similar, we get
$$
\|\proj z\proj q\|=\|zq\|\frac{\|\proj p\proj q\|}{\|pq\|}=\|zq\|\lambda /\Lambda.
$$

Let $u$ be the intersection point of $\overrightarrow{zp}$ and the vertical line through $d$.
Since the triangles $pzq$ and $pud$ are similar, we get
$$
\|\proj z\proj q\|=\lambda/\Lambda\cdot \frac{\|pz\|\|du\|}{\|pu\|}.
$$

Here, we have $\|pz\|\leq 2$, $\|du\|=O(1)$, and $\|pu\|\geq \sqrt 2(\Lambda-1)>\Lambda-1$, so we get $\|\proj z\proj q\|=O(\lambda/\Lambda^2)$.
This means that at the line $\ell_{\proj o}$, the ray $\overrightarrow{pd}$ is contained in the slab $S(\proj p,\alpha,\eps)$, for $\eps=O(\lambda/\Lambda^2)$.
The maximum deviation is realized when $\proj p$ is on the left or right edge of $\proj B$, and then the vertical distance between $\overrightarrow{pd}$ and $\proj k$ is maximal on the vertical line through the right or left edge of $\proj B$, respectively, which is at most twice the deviation at the center line $\ell_{\proj o}$.
Hence, the ray $\overrightarrow{pd}$ is contained in the slab $S(\proj p,\alpha,\eps)$ for $\eps=O(\lambda/\Lambda^2)$ in all of~$\proj B$.
\end{proof}

Define the endpoints $\bar a_\sigma,\bar b_\sigma$ of the variable segments in the gadget such that $r_\sigma=\bar a_\sigma \bar b_\sigma$ with $\bar a_\sigma$ to the left. 
Define the relative origin of the gadget as $\bar o\mydef c_l^r$, which is the right lever corner responsible for the segment $r_l$.
We place the gadget so that the point $\bar o$ is on the ray $\overrightarrow{\proj od_1}$ with $x(\bar o)=x(d_1)+1$.

Let us introduce a grid of slabs parallel to $\alpha$ and $\beta$.
Recall that there are variable segments $s_i$ for a subset of the indices $i\in\{1,\ldots,N\}$.
In the following we do not care whether a particular segment $s_i$ exists or not.
We make the construction such that even if they all exist, the construction works, and then it will also work when some are missing.

We will later (using Lemma~\ref{lemma:proj}) choose a constant $C$, and using that, we define
\[ \nu\mydef \frac{13.5}{C N^2}, \ \rho\mydef \frac{\nu}{9}=\frac{1.5}{C N^2}\text{, and }\eps\mydef \frac{\rho}{48} = \frac{1}{32 C N^2}.\]
For each $\sigma\in\{1,\ldots,N\}$ and $\gamma \in \{0,1,2,3\}$ we define a slab
$$L^{\gamma}_\sigma \mydef S(\proj o+(0,(\sigma-1) \nu + \gamma \rho),\alpha,\varepsilon),$$
which we denote as an \emph{$L$-slab}.
Let $\tau(i)\mydef 2$, $\tau(j)\mydef 1$, and $\tau(l)\mydef 0$. For each $\sigma\in\{i,j,l\}$ and $\gamma \in \{0,1,2,3\}$ we define a slab
$$R^{\gamma}_\sigma \mydef S(\,\proj o+(0,\tau(\sigma) \nu + \gamma \rho),\beta,\varepsilon),$$
which we denote as an \emph{$R$-slab}.
See Figure~\ref{fig:slabs} for an illustration of the area where the $L$-slabs intersect the $R$-slabs.

In the case of gadgets with just two guard segments $r_i,r_j$, we define the relative origin of the gadget to be $\bar o\mydef c_j^r$, which we place on the ray $\overrightarrow{\proj o d_1}$ with $x(\bar o)=x(d_1)+1$,
and we define $\tau(i)\mydef 1$ and $\tau(j)\mydef 0$.
Then, the $R$-slabs $R_\sigma^\gamma$ are defined as above for $\sigma\in\{i,j\}$.

\begin{figure}[htbp]
	\centering
	\includegraphics[page=9]{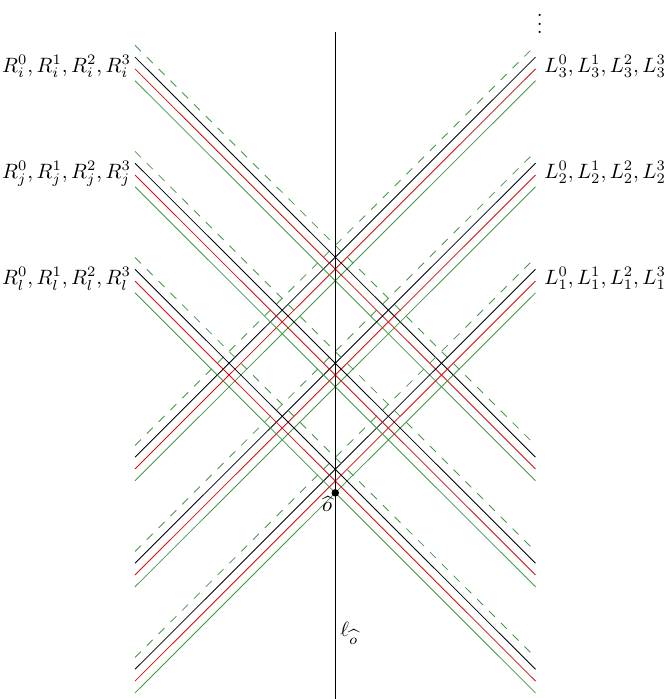}
	\caption{A part of the grid of $L$- and $R$-slabs.
	Each slab is one line in the grid.
	Full green lines are $L_\sigma^0$ and $R_\sigma^0$.
	Red lines are $L_\sigma^1$ and $R_\sigma^1$.
	Black lines are $L_\sigma^2$ and $R_\sigma^2$.
	Dashed green lines are $L_\sigma^3$ and $R_\sigma^3$.
}
	\label{fig:slabs}
\end{figure}

Let $\proj V$ be the square $\proj{o}+[-N\nu,N\nu]\times [-N\nu,N\nu]$.
Let us denote by $\mathcal{L}$ all the rays with endpoints at a propagation corner $c_\sigma^s$, $\sigma\in\{1,\ldots,N\}$, going through $d_0$ or $d^0$, and rays with endpoints $a_\sigma$ or $b_\sigma$ through the point $d_0$.
We will use Lemma~\ref{lemma:proj} to ensure that all rays in $\mathcal{L}$ are inside a predefined slab within the square~$\proj V$, which is made precise by the following lemma.

\begin{lemma}\label{lemma:slabs}
We can choose $C$ large enough such that for any corridor to be attached to the polygon $\poly$, the following properties are satisfied.
\begin{enumerate}
\item\label{lslabs:0}
The intersection of any $L$-slab with the line $\ell_{\proj o}$ is contained in $V$.
\item\label{lslabs:2}
For each $\sigma\in\{1,\ldots,N\}$, it holds that $\overrightarrow{c_\sigma^s d_0}\cap V\subset L^{0}_\sigma,\overrightarrow{a_\sigma d_0}\cap V\subset L^{1}_\sigma, \overrightarrow{b_\sigma d_0}\cap V\subset L^{2}_\sigma$, and $\overrightarrow{c_\sigma^s d^0}\cap V\subset L^{3}_\sigma$.
\end{enumerate}
\end{lemma}

\begin{proof}
The first property follows since the centers of the slabs intersect $\ell_{\proj o}$ at the points $\proj o,\proj o+(0,\nu),\ldots,\proj o+(0,(N-1)\nu)$.

To prove the second property, we use Lemma~\ref{lemma:proj}, where the square $V\mydef c_1^s+[-13.5N,13.5N]\times [-13.5N,13.5N]$ plays the role as the unit square $B$ and $\proj V$ plays the role as $\proj B$.
In order to apply the lemma, we therefore scale $V$, $\proj V$, and all other objects by a factor of $1/27N$.
When using the lemma for the first three rays, we use $d_0$ as $d$, and for the first last, we use $d^0$.
Recall that each variable segment in the base pockets has a length in the interval $[(1-\tilt)3/2,3/2]$ with $\tilt\mydef 1/32=\eps CN^2$, and that the center points of the variable segments in each pocket are on a vertical line.
We get from the lemma that choosing $C$ large enough, property~\ref{lslabs:2} follows.
\end{proof}

Consider a gadget containing the variable segments $r_\sigma\mydef \bar a_\sigma\bar b_\sigma$, $\sigma\in\{i,j,l\}$.
Recall that the relative origin of a gadget is the point $\bar o\mydef c_l^r$, which is placed on the ray $\overrightarrow{\proj o d_1}$ with $x(\bar o)=x(d_1)+1$.
We state properties of gadgets with three variable segments $r_i,r_j,r_l$, but there are natural analogues for gadgets with just two variable segments $r_i,r_j$.
The following observation, which is analogous to Observation~\ref{obs:basePocketLines}, will be clear from the construction of each gadget.

\begin{observation}\label{obs:gadgetLines}
For each $\sigma\in\{i,j,l\}$, the following holds:
\begin{itemize}
\item
$c_\sigma^r$ is on the line through $\bar o+(-\tau(\sigma)\nu,0)$ with direction $\beta$,

\item
$\bar b_i$ is on the line through $\bar o+(-\tau(\sigma)\nu-\rho,0)$ with direction $\beta$,

\item
$\bar a_i$ is on the line through $\bar o+(-\tau(\sigma)\nu-2\rho,0)$ with direction $\beta$.
\end{itemize}
\end{observation}
 
We now get the following from lemma.

\begin{lemma}\label{lemma:r-slabs}
We have the following two properties.
\begin{enumerate}
\item\label{rslabs:1} The intersection of any $R$-slab with the line $\ell_{\proj o}$ is contained in $V$.
\item\label{rslabs:2} For each $\sigma\in\{i,j,l\}$, it holds that $\overrightarrow{c_\sigma^r d_1}\cap V\subset R^{0}_\sigma,\overrightarrow{\bar b_\sigma d_1}\cap V\subset R^{1}_\sigma,\ \overrightarrow{\bar a_\sigma d_1}\cap V\subset R^{2}_\sigma$, and $\overrightarrow{c_\sigma^r d^1}\cap V\subset R^{3}_\sigma$.
\end{enumerate}
\end{lemma}

\begin{proof}
The first property follows since the centers of the slabs intersect $\ell_{\proj o}$ at the points $\proj o,\proj o+(0,\nu),\proj o+(0,2\nu)$.

To prove the second property, we again use Lemma~\ref{lemma:proj}, but this time reflected along a vertical axis.
The square $V\mydef c_l^r+[-N\nu,N\nu]\times [-N\nu,N\nu]$ plays the role as the unit square $B$, and $\proj V$ will play the role as $\proj B$, so in order to apply the lemma, we scale $V$, $\proj V$, and all other objects by a factor of $1/2N\nu=\frac{C N}{27}$.
For the first three rays (that pass through $d_1$), we use $d$ in the lemma as $d_1$.
For the last one, we use $d^1$.
We then get from the lemma that choosing $C$ large enough, property~\ref{lslabs:2} follows.
\end{proof}

We can now conclude the following, which is obvious from property~\ref{lslabs:0} of Lemma~\ref{lemma:slabs} and property~\ref{rslabs:1} of Lemma~\ref{lemma:r-slabs}, since the directions $\alpha$ and $\beta$ of the slabs have slopes $1$ and $-1$, respectively.

\begin{lemma}
The intersection of any $L$-slab with any $R$-slab is contained in $\proj V$.
\end{lemma}

From Lemmas~\ref{lemma:slabs} and~\ref{lemma:r-slabs}, we get that for each $\sigma\in\{i,j,l\}$, only the lever corners $c_\sigma^s$ and $c_\sigma^r$ can see the critical segment $f_\sigma$ of the pair $s_\sigma,r_\sigma$.
This is because the rays from other lever corners through the door endpoints $d_0,d^0,d_1,d^1$ are contained in the wrong slabs by the lemmas.
The vertical edge above the corner $d^0$ blocks the visibility from corners $c_{\sigma'}^s$ for $\sigma'<\sigma$, and the vertical edge below $d_0$ blocks the visibility from corners $c_{\sigma'}^s$ for $\sigma'>\sigma$, and unwanted visibilities from marked corners in the gadget are likewise blocked by the edges incident at $d^1$ and $d_1$.
It now remains to complete the specification of the corridor and the gadget so that indeed the corners  $c_\sigma^s$ and $c_\sigma^r$ can see $f_\sigma$ and so that the instance satisfies the \mrcc\ promise.

\subsubsection{Specification of the corridor}
We are now ready to describe the exact construction of the corridor.
As mentioned before, the bottom wall is simply the horizontal line segment $d_0d_1$.
We first describe the approximate shape of the upper wall, defined by a polygonal curve $\Psi$ connecting the points $d^0$ and $d^1$; see Figure~\ref{fig:slabs2}.
We will then add some helper spikes to get the final polygonal curve $\Psi'$.

\begin{figure}[htbp]
	\centering
	\includegraphics[page=20, trim=5cm 5cm 5cm 5cm, clip]{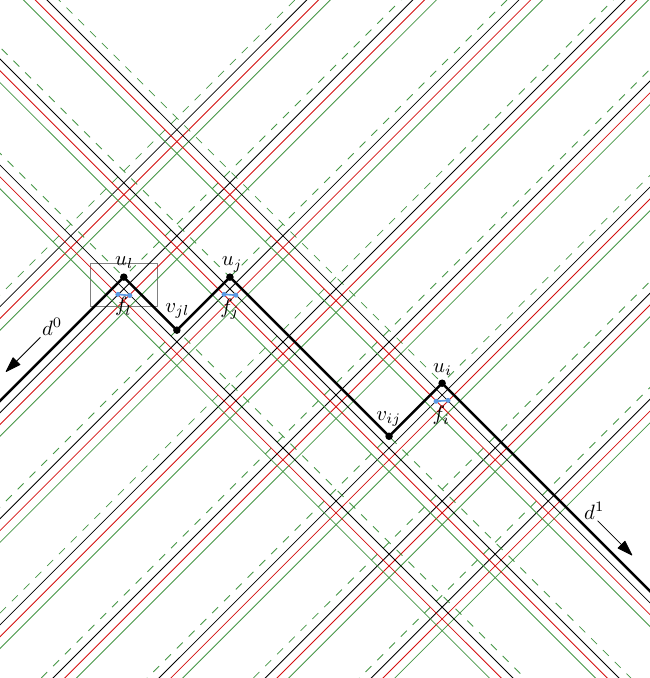}
	\caption{The construction of the upper corridor wall.
	The color scheme of the slabs is as in Figure~\ref{fig:slabs}.
	The critical segments are blue.
	See Figure~\ref{fig:slabs3} for a closeup of the gray rectangle.}
	\label{fig:slabs2}
\end{figure}

Note that in the corridor construction here we assume that $i < j < l$.
In particular, the $L$-slabs $L^\gamma_l$ are above the $L$-slabs $L^\gamma_j$, which are above $L^\gamma_i$.
For the $R$-slabs it is the opposite, i.e., the $R$-slabs $R^\gamma_l$ are below the $R$-slabs $R^\gamma_j$, which are below $R^\gamma_i$.

Figure~\ref{fig:slabs2} shows the grid of slabs and a sketch of the curve $\Psi$ approximating the upper wall (excluding most of the leftmost and rightmost edge of $\Psi$, with endpoints at $d^0$ and $d^1$, respectively, since they are too long to be pictured together with the middle segments).
For $\sigma\in\{i,j,l\}$, let $u_\sigma$ be the intersection point of the rays $\overrightarrow{c_\sigma^s d^0}$ and $\overrightarrow{c_\sigma^r d^1}$.
Let $v_{ij}$ be the intersection point of the rays $\overrightarrow{c_i^s d^0}$ and $\overrightarrow{c_j^r d^1}$, and $v_{jl}$ the intersection point of the rays $\overrightarrow{c_j^sd^0}$ and $\overrightarrow{c_l^rd^1}$.
The curve $\Psi$ is then a path defined by the points $d^0u_iv_{ij}u_jv_{jl}u_ld^1$.
By Lemma~\ref{lemma:slabs}, the set $\Psi \cap \proj V$ is contained in the union of the $L$-slabs and the $R$-slabs, as shown in Figure~\ref{fig:slabs2}.
Due to the relative position of the slabs $L^\gamma_l, L^\gamma_j, L^\gamma_i$ and $R^\gamma_l, R^\gamma_j, R^\gamma_i$ as discussed above, the curve $\Psi$ is $x$-monotone, and the point $v_{ij}$ (resp.~$v_{jl}$) has smaller $y$-coordinate than the neighbouring points $u_i,u_j$ (resp.~$u_j,u_l$), i.e., the curve $\Psi$ always has a zig-zag shape resembling the one from Figure~\ref{fig:slabs2}.
We therefore get that when adding marked rectangles around the restricted critical segments as described in the following, the vertical and horizontal bars are pairwise disjoint, satisfying the bar intersection promise.

\begin{figure}[h!]
	\centering
	\includegraphics[page=21]{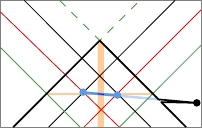}
	\caption{A closeup of the construction of a helper spike.
	The figure also shows the restricted critical segment (dark blue) and the horizontal and vertical bars (orange).}
	\label{fig:slabs3}
\end{figure}

We obtain the final upper wall $\Psi'$ by adding helper spikes, as shown in Figure~\ref{fig:slabs3}:
We consider the extension of each critical segment in the direction of decreasing $y$-coordinates.
We hit a segment of $\Psi$, and make a helper spike there, with a corner that can see the region below the critical segment.
This ensures that the created helper corner is the bottom endpoint of both edges of the helper spike, thus satisfying the invisibility invariant.

\subsection{Adjusting propagation spikes}\label{sec:propSpikes}

Knowing how the corridor is constructed, as described in the previous section, it is now possible for us to describe how we choose the edges incident at the propagation corners in the base pockets.
Consider a variable segment $s_i=a_ib_i$ and the associated propagation corner $c_i^s$, and let $r_i$ be the variable segment in a gadget which is the partner of $s_i$; see Figure~\ref{fig:adjustingProp}.
We have so far not described exactly how to choose the upper edge of $\poly$ incident at $c_i^s$, only that it would be a corner close to $c_i^s+(1,1)$.

We need our construction to be sufficiently flexible in the sense that for every value of a variable in $X$, we can construct a cover of triangles that specifies that value on all variable segments.
But recall that each variable is restricted to a range of length $\delta\mydef n^{-7}$, where $n$ is the number of variables in the formula $\Phi$.
In our construction, we therefore only need the flexibility corresponding to a tiny subsegment $s_i'=a_i'b_i'$ contained in $s_i$, and likewise, we only need to cover a tiny part $r_i'$ of $r_i$, where $s_i'$ and $r_i'$ cover fractions of size just $O(n^{-7})$ of the full segments $s_i$ and $r_i$.
Let $f_i'$ be the part of the critical segment $f_i$ corresponding to the restricted parts $(s_i',r_i')$.
The full critical segment $f_i$ has length $\Theta(1/N^2)=\Theta(n^{-6})$, so the restricted $f_i'$ has length $\|f_i'\|=\Theta(\delta\|f_i \|)=\Theta(n^{-13})$.

Let the endpoins of $f_i$ and $f_i'$ be $\kappa_0,\kappa_1,\kappa_0',\kappa_1'$, such that $f_i=\kappa_0\kappa_1$ and $f'_i=\kappa_0'\kappa_1'$, with $\kappa_0$ and $\kappa_0'$ to the left.
We consider the ray $\xi\mydef \overrightarrow{a_i'd_0}$, which passes through the right endpoint $\kappa_1'$ of $f_i'$.
Let $e$ be the point on $\xi$ such that $y(e)=y(\kappa_1')+1/10CN^2$.
We now choose the upper segment incident at $c_i^s$ such that the extension of the edge contains the point $e$.
We need to argue that with this definition, the propagation corner $c_i^s$ can see all points of $f'_i$, since then a lever piece can represent any value in the restricted range.
For this, it is enough to see that a triangle contained in $\poly$ can cover $c_i$ and all of $f_i'$.
We therefore calculate how large a portion of $f_i$ can be covered by a triangle with a lever edge contained in $\xi$.
The orange triangle $T=c_i^sge$ in the figure is the maximal such piece, and we need to find a lower bound on the length of $T\cap f_i$.
Let $F$ be the intersection of $T$ and the line through $g$ parallel to $f_i$.
Then $\|F\|=\Theta(1)$.
We also have $y(e)-y(g)=\Theta(N^2)$ and $y(e)-y(\kappa_1')=\Theta(1/N^2)$.
Since the parts of $T$ above $F$ and above $f_i$ are similar, we get that $\|T\cap f_i\|=\Theta(N^{-4})=\Theta(n^{-12})$.
Since $f_i'$ has length $\Theta(n^{-13})$, it follows that $f_i'\subset T$ when $n$ is large enough.

\begin{figure}[h!]
	\centering
	\includegraphics[page=22]{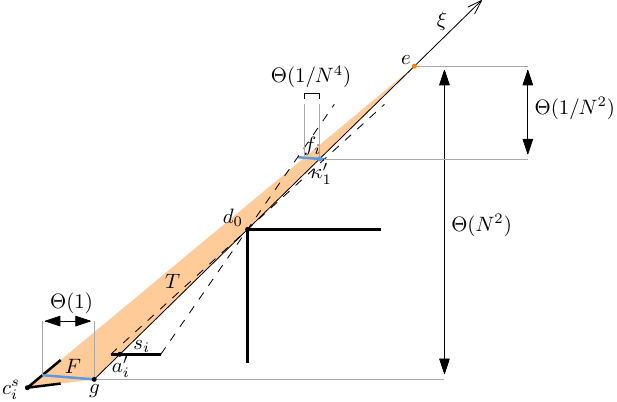}
	\caption{Figure showing the principle for choosing the upper incident edge of the propagation corner $c_i^s$.
	The figure is not to scale.}
	\label{fig:adjustingProp}
\end{figure}

We adjust the edges incident at the lever corners $c_\sigma^r$ in the gadgets in a similar way.
We can now state the lemma expressing that the corridor works as intended.

\begin{lemma}\label{lem:corridor_works}
If the constructed instance $\II_2$ has a cover, the following holds.
For $\sigma\in\{i,j,l\}$, if $s_\sigma$ and $r_\sigma$ are both right-oriented, then $\val{s_\sigma}{c_\sigma}\leq\val{r_\sigma}{c_\sigma^r}$.
If both segments are left-oriented, then $\val{s_\sigma}{c_\sigma}\geq\val{r_\sigma}{c_\sigma^r}$.

For each $\sigma\in\{i,j,l\}$, consider any value $y_\sigma\in I(x_\sigma)$ in the restricted range of the variable $x_\sigma$.
There exists $9$ triangles that cover the lever corners $c_\sigma$ and $c_\sigma^r$, $\sigma\in\{i,j,l\}$, and the helper corners in the corridor such that the following holds.
The triangles cover the $3$ marked rectangles (containing the restricted critical segments $f'_\sigma$) in the corridor and we have $y_\sigma=\val{s_\sigma}{c_\sigma}=\val{r_\sigma}{c_\sigma^r}$.
\end{lemma}

\begin{proof}
For the first part, for each $\sigma\in\{i,j,l\}$, we consider the restricted critical segment $f'_\sigma$ in the corridor corresponding to the pair $(s'_\sigma,r'_\sigma)$.
It follows from Lemma~\ref{lemma:slabs} that among all propagation corners in the base pockets, only the intended corner $c_\sigma$ can see $f'_\sigma$.
Likewise, by Lemma~\ref{lemma:r-slabs}, only the corner $c_\sigma^r$ in the gadget can see $f'_\sigma$.
Hence, the pieces covering $c_\sigma$ and $c_\sigma^r$ must together cover $f'_\sigma$.
The statement therefore follows from Lemma~\ref{lemma:copyLemma}.

On the other hand, it follows from the discussion above (in this section) that for any $y_\sigma\in I(x_\sigma)$ in the restricted range of $x_\sigma$, there are two triangles $T_0,T_1$ covering $c_\sigma$ and $c_\sigma^r$ that specify the value $y_\sigma$ at $s_\sigma$ and $r_\sigma$, respectively, and so that $T_0$ and $T_1$ together cover $f'_\sigma$.
Using $3$ triangles covering the helper corners in the corridor, we get a cover for the marked rectangles in the corridor.
\end{proof}

\subsection{The $\ge$-inequality gadget}\label{sec:Geinequality}

The $\ge$-inequality gadget is shown in Figure~\ref{fig:gePrinciple}.
In this gadget, $r_i$ is left-oriented and $r_j$ is right-oriented.
In addition to the lever corners $c_i^r$ and $c_j^r$, we have a lever corner $c$ which is the left-responsible for $r_i$ and $r_j$.
Using Observations~\ref{obs:ineq1} and~\ref{obs:ineq2}, we get that $\val{r_i}{c_i^r}\geq \val{r_i}{c} \geq \val{r_j}{c}\geq\val{r_j}{c_j^r}$.
Recall the bar intersection promise described in Section~\ref{sec:mrcc}.
In order to satisfy this promise, we consider the vertical bar of the marked rectangle of the restricted range of $r_i$ and the horizontal bar of that of $r_j$, and we make the intersection of these two bars a marked rectangle $R_h$.
We also make a marked corner $h$ such that $h$ and $R_h$ can be covered by a single piece.

\begin{figure}[h!]
\centering
\includegraphics[page=23]{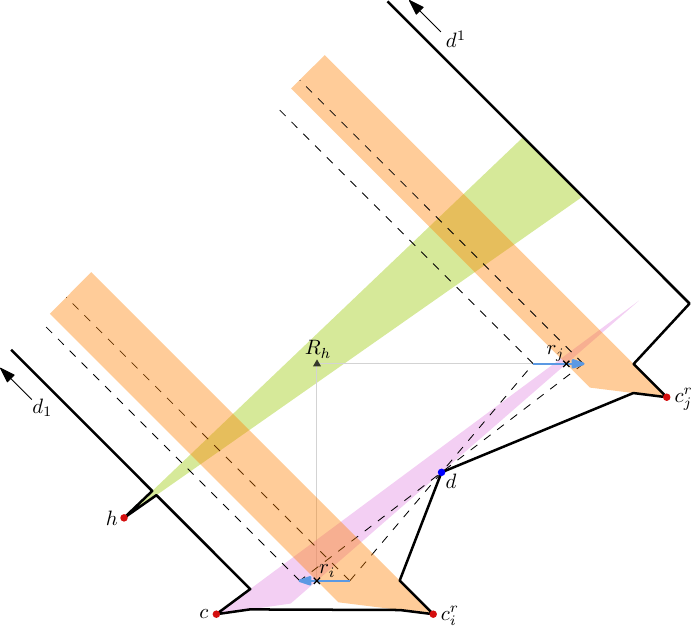}
\caption{The principle of the $\ge$-inequality gadget.
The crosses at the variable segments $r_i$ and $r_j$ mark the position of the infinitesimal restricted range of the segments.
The infinitesimal marked rectangle $R_h$ is shown as a triangular point, which can be covered by a piece covering the marked point $h$.
}
\label{fig:gePrinciple}
\end{figure}

\subsection{The $\le$-inequality gadget}\label{sec:Leinequality}

The $\le$-inequality gadget is identical to the $\ge$-inequality gadget, except that the orientations of the segments $r_i$ and $r_j$ are swapped.

\subsection{The $\ge$-addition gadget}\label{sec:GEadditionGadget}

In this section we present the construction of the $\geq$-addition gadget which represents an inequality $x_i+x_j\ge x_l$, where $i,j,l\in\{1,\ldots,n\}$.
In Section~\ref{sec:LEadditionGadget} we show how to modify the construction to obtain the $\leq$-addition gadget for the inequality $x_i+x_j \le x_l$.
Recall that in Section~\ref{sec:bottomWall}, we rewrote all equations of the form $x_i+x_j = x_l$ in $\Phi$ as two inequalities.
We will then add a gadget for each inequality to our polygon $\poly$.

\subsubsection{Idea behind the gadget construction}\label{subsec:addition-idea}
The principle behind the gadget follows one from~\cite{abrahamsen2018art}.
See Figure~\ref{fig:additionPrinciple} (left) for a sketch of the construction.

\begin{lemma}[\cite{abrahamsen2018art}]\label{lem:additionPrin}
For some $v>0$, let $e_i\mydef (v,0),e_j\mydef (-v,0),e_l\mydef (0,0)$.
For some $h\in\R$, let $\ell_h$ and $\ell_{2h}$ be the horizontal lines with $y$-coordinate $-h$ and $-2h$, respectively.
Let $\gamma$ be a point with $y(\gamma)\neq 0$.
For $\sigma\in\{i,j,l\}$, let $p_\sigma\mydef \overleftrightarrow{\gamma e_\sigma}\cap \ell_h$ and $p'_\sigma\mydef \overleftrightarrow{\gamma e_\sigma}\cap \ell_{2h}$.
Then $x(p_l)=\frac{x(p_i)+x(p_j)}2$ and $x(p'_l)=x(p_i)+x(p_j)$.
\end{lemma}

\begin{figure}[h!]
\centering
\includegraphics[page=1]{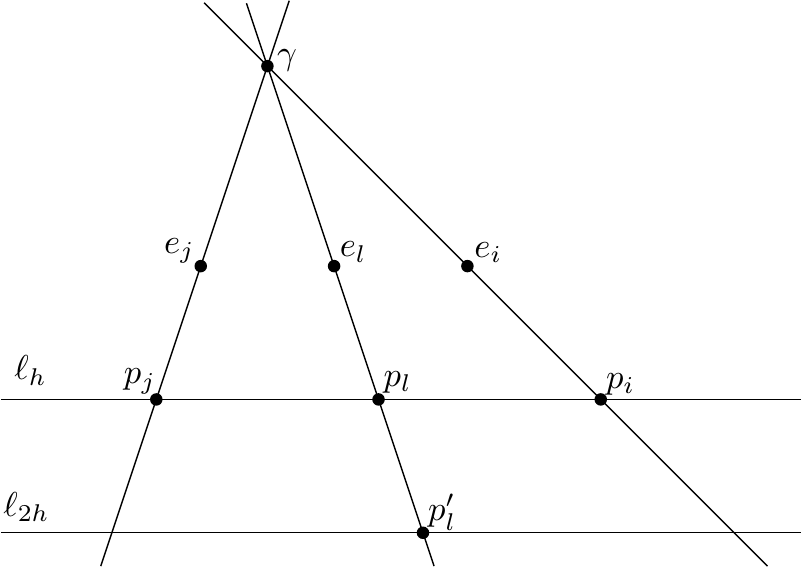}\quad
\includegraphics[clip, trim = 0cm 1cm 0cm 0cm,width=0.45\textwidth]{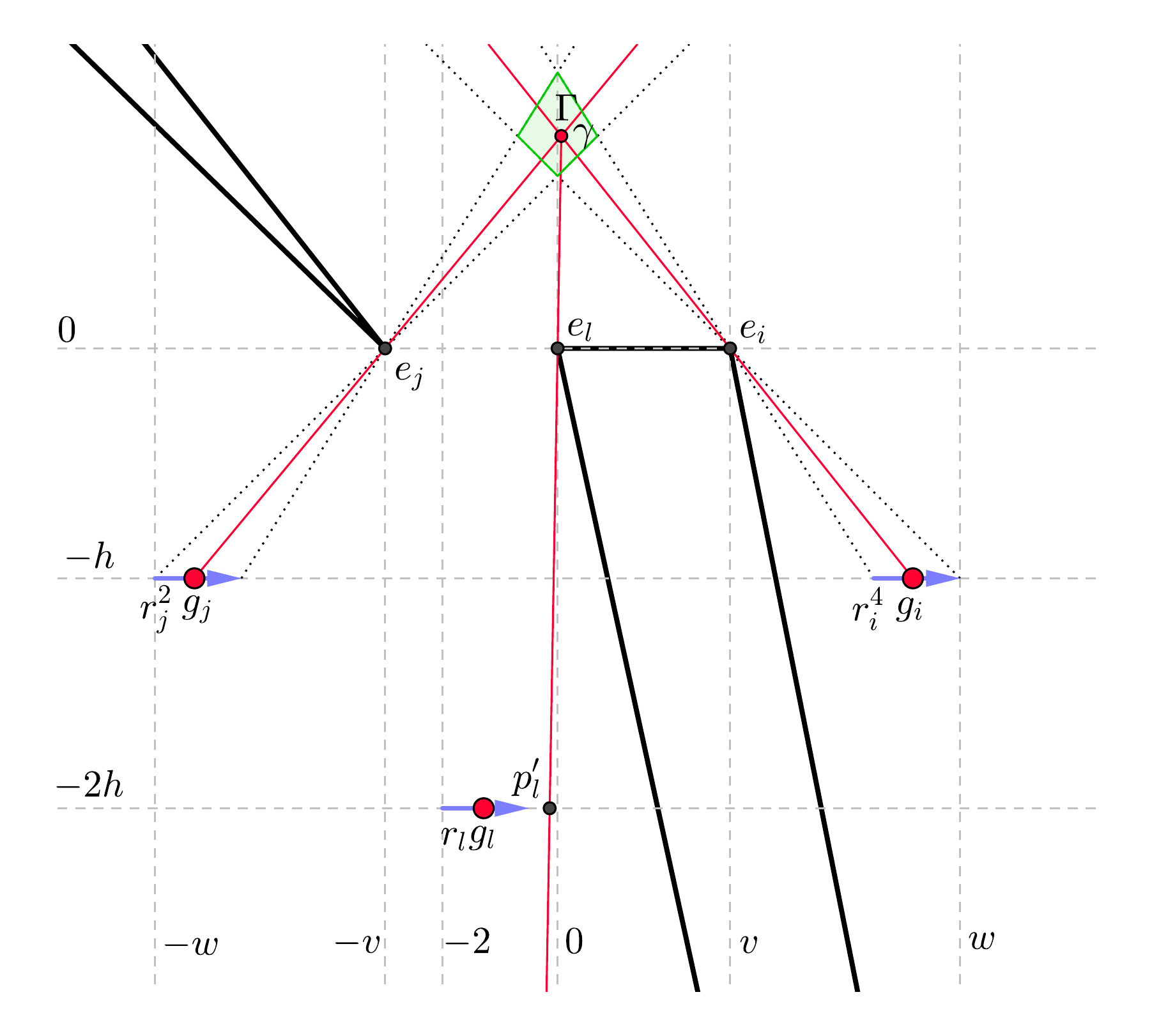}
\caption{Left: The construction described in Lemma~\ref{lem:additionPrin}.
The point $e_l$ is at the origin, and then the $x$-coordinate of $p'_l$ is the sum of those of $p_i$ and $p_j$.
Right: How we place variable segments on the lines $\ell_h$ and $\ell_{2h}$, and how we use walls of $\poly$ to restrict visibility.}
\label{fig:additionPrinciple}
\end{figure}

We use the lemma in the following way; see Figure~\ref{fig:additionPrinciple} (right).
For some value $w>v+3/2$, we place two variable segments $r^4_i$ and $r^2_j$ at the line $\ell_h$, such that the right endpoint of $r^4_i$ is $(w,-h)$ and the left endpoint of $r^2_j$ is $(-w,-h)$.\footnote{We use superscripts for $r^4_i$ and $r^2_j$, since we will also have variable segments $r^1_i,r^2_i,r^3_i$ and $r^1_j$ representing the same variable. More details will be given later.}
We then place a variable segment at $r_l$ with left endpoint at $(-2,-2h)$.
It then follows from the lemma that if the points $p_i,p_j,p'_l$ are on their respective segments $r^4_i,r^2_j,r_l$, then $p'_l$ represents a value which is the sum of those of $p_i$ and $p_j$ (if the points are not on the segments, this will still hold if we extrapolate the correspondence between points and values in the natural way).


Let $\Gamma$ be a collection of points $\gamma$ such that the ray $\overrightarrow{\gamma e_i}$ intersects $r^4_i$, and the ray $\overrightarrow{\gamma e_j}$ intersects $r^2_j$.
Then $\Gamma$ is a quadrilateral, bounded by a segment on each of the following rays: the rays with origin at the endpoints of $r^4_i$ and containing $e_i$, and the rays with origin at the endpoints of $r_j$ and containing $e_j$.
As shown in Figure~\ref{fig:additionPrinciple} (right), we use walls of $\poly$ to create a kind a pivots, so that we get the following properties. \begin{itemize}
\item
For every point $g_i$ on $r^4_i$ and $\gamma$ in $\Gamma$, the points $\gamma$ and $g_i$ can see each other if and only if $\gamma$ is on or to the right of the line $\overleftrightarrow{g_i e_i}$.

\item
For every point $g_j$ on $r^2_j$ and $\gamma$ in $\Gamma$, the points $\gamma$ and $g_j$ can see each other if and only if $\gamma$ is on or to the right of the line $\overleftrightarrow{g_j e_j}$.

\item
For every point $g_l$ on $r_l$ and $\gamma$ in $\Gamma$, the points $\gamma$ and $g_l$ can see each other if and only if $\gamma$ is on or to the left of the line $\overleftrightarrow{g_l e_l}$.
\end{itemize}

We then get from Lemma~\ref{lem:additionPrin} that if points $g_i,g_j,g_l$ on $r_i^4,r_j^2,r_j$, respectively, together see $\Gamma$, then the value represented by $g_l$ is at most the sum of the values represented by $g_i$ and $g_j$.
We are going to make a marked rectangle around (a restricted part of) $\Gamma$, which must be covered by pieces covering lever corners that are responsible for the segments $r_i^4,r_j^2,r_j$.

\begin{figure}
\centering
\includegraphics[width=\textwidth]{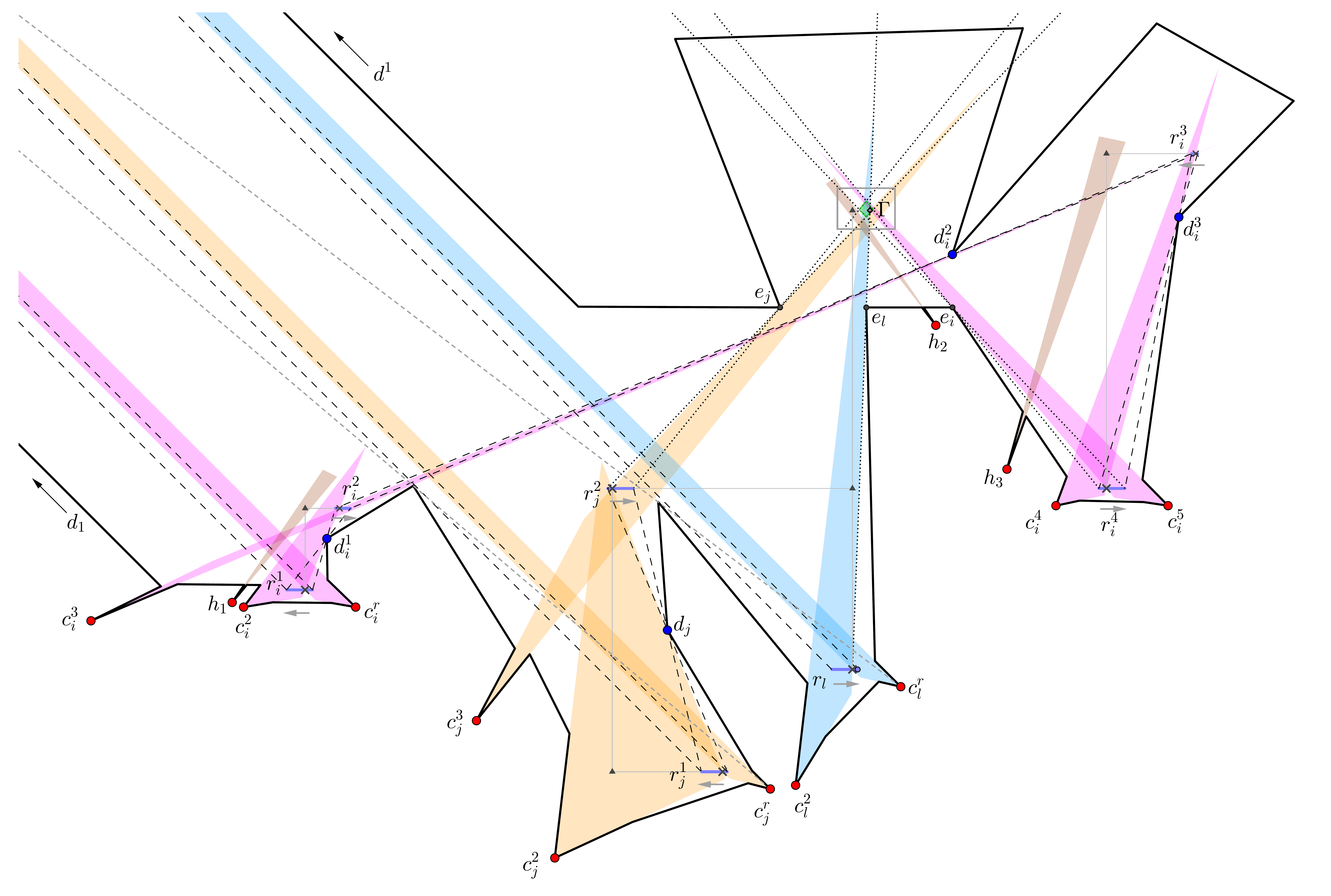}
\caption{The $\geq$-addition gadget.
The gray arrows show the orientations of the variable segments.
Red points are lever corners and blue are pivots.
The crosses mark the infinitesimal restricted ranges on the variable segments, and the diamond-shaped point in $\Gamma$ is the infinitesimal restricted range $\Gamma'$ of $\Gamma$.
Violet regions are lever triangles covering variable segments representing $x_i$, orange represent $x_j$, and blue represent $x_l$.
The gray dashed rays from lever corners $c_j^r$ and $c_l^r$ show that the corners cannot see $r_i^2$ and $r_j^2$, respectively, so that the bi-cover property is satisfied.
The triangular points indicate the extra marked rectangles we add in order to satisfy the bar intersection promise.
These can be covered by pieces covering the marked corners $h_1,h_2,h_3, c_j^2,c_l^2$.
For a closeup of the region in the gray rectangle, see Figure~\ref{fig:additionGadgetCloseup}.
}
\label{fig:additionGadget}
\end{figure}

\subsubsection{Actual gadget}

We use the above principle with values $h\mydef 10.5/CN^2$, $v\mydef 5/CN^2$, and $w\mydef 13.5/CN^2$.
The complete gadget can be seen in Figure~\ref{fig:additionGadget}.
As shown in the figure, we need the variable segments $r_i=r_i^1$ and $r_j=r_j^1$ to be left-oriented, whereas $r_l$ is right-oriented.
We apply sequences of lever mechanisms in order to create the inequalities
\begin{align*}
\val{r_i^1}{c_i^r} & \geq \val{r_i^1}{c_i^2} \geq \val{r_i^2}{c_i^2}\geq \val{r_i^2}{c_i^3}\geq \val{r_i^3}{c_i^3}\geq \val{r_i^3}{c_i^4}\geq \val{r_i^4}{c_i^4}\geq \val{r_i^4}{c_i^5}, \\
\val{r_j^1}{c_j^r} & \geq\val{r_j^1}{c_j^2}\geq\val{r_j^2}{c_j^2}\geq\val{r_j^2}{c_j^3}, \quad \text{and} \\
\val{r_l}{c_l^r} & \leq\val{r_l}{c_l^2}.
\end{align*}

Just as we define the instance $\II_2$ such that we only need to cover the infinitesimal restricted ranges of $r_i^4$, $r_j^2$, and $r_l$, we also define an infinitesimal restricted region $\Gamma'$ of $\Gamma$.
We get $\Gamma'$ as the intersections of all rays from the restricted ranges of $r_i^4$ and $r_j^2$ through $e_i$ and $e_j$, respectively.
Thus, $\Gamma'$ is a very small quadrilateral contained in $\Gamma$.
We then make the axis-parallel bounding box $R_{\Gamma'}$ a marked rectangle of $\II_2$; see the closeup of the region around $\Gamma$ in Figure~\ref{fig:additionGadgetCloseup}.

We need the extra variable segments since the segments $r^1_i,r^1_j,r_l$ (that are connected to the base pockets via the corridor) do not have the right positions to play the roles of the segments $r_i^4,r_j^2,r_l$ described above.
For instance, in the case of $r_j^1$, the corner $c_j^r$ is the right-responsible for $r_j^1$, so the other responsible corner $c_j^2$ must be the left-responsible.
However, the corner responsible for covering $r^2_j$ and part of $\Gamma'$ must be the right-responsible for $r^2_j$ for the construction to work as described in Section~\ref{subsec:addition-idea}.
Therefore, $r_j^1$ cannot be used as $r_j^2$, but this is solved by connecting $r_j^1$ and $r_j^2$ by lever mechanisms.

As in the inequality gadgets, Section~\ref{sec:Geinequality} and~\ref{sec:Leinequality}, we introduce extra marked rectangles (denoted as triangles in the figures) to satisfy the bar intersection promise, and we introduce helper corners that can see them.

To see that the gadget works as intended, suppose that 
%
%
pieces $\piece_i,\piece_j,\piece_l$ cover $c_i^5,c_j^3,c_l^2$, respectively, and that they together cover $R_{\Gamma'}$.
Let $g_i\mydef  \pp{r_i^4}{c_i^5}$, $g_j\mydef \pp{r_j^2}{c_j^3}$, and $g_l\mydef \pp{r_l}{c_l^2}$.
As $\piece_i,\piece_j,\piece_l$ cover $R_{\Gamma'}$ and thus $\Gamma'$, it follows that $g_i,g_j,g_l$ together see $\Gamma'$.
By the remarks in Section~\ref{subsec:addition-idea}, we then get $\val{r_i^4}{c_i^5}+\val{r_j^2}{c_j^3} \geq \val{r_l}{c_l^2}$, as desired.


\begin{figure}
\centering
\includegraphics[width=0.5\textwidth]{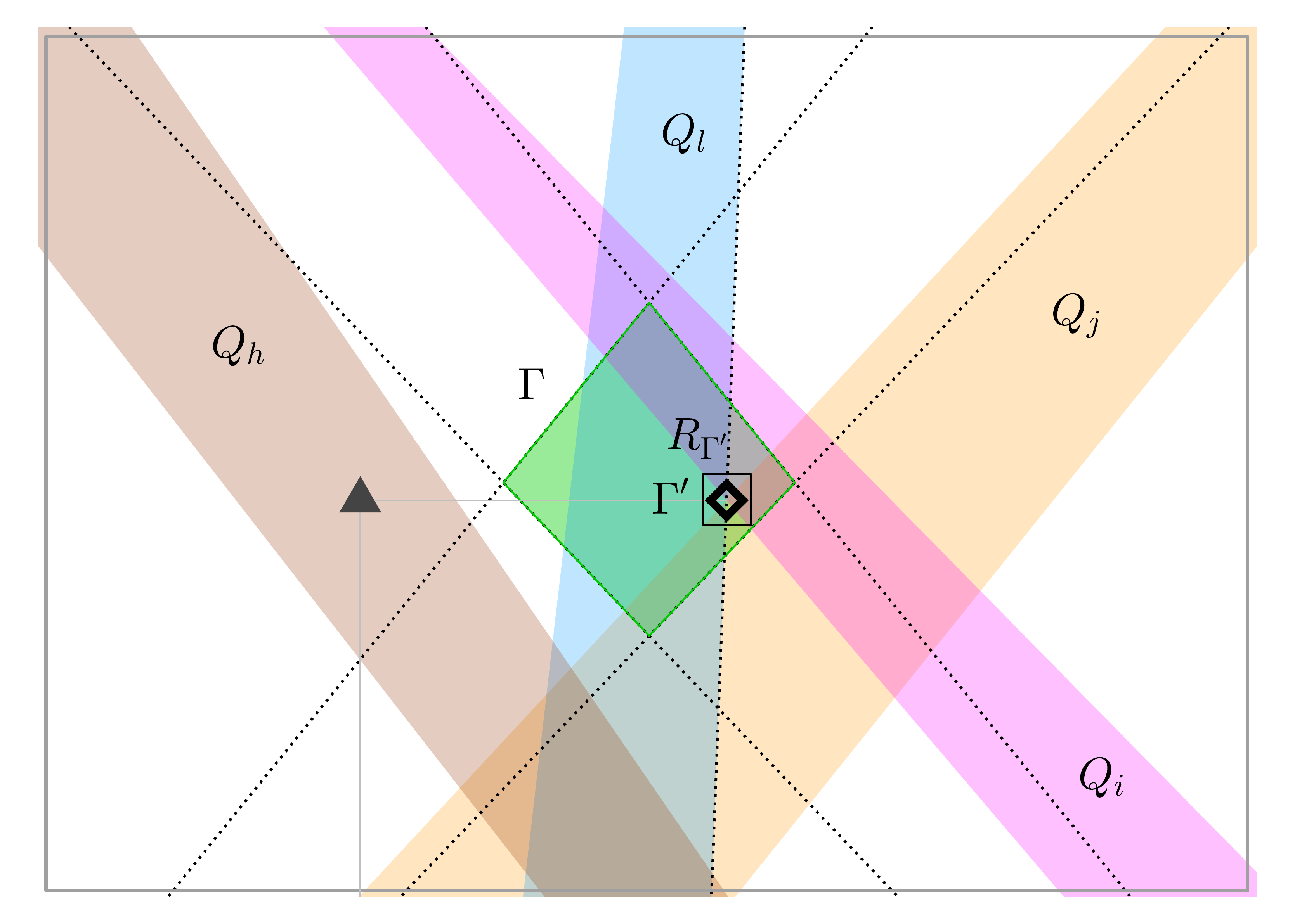}
\caption{Closeup of the region in the gray rectangle from Figure~\ref{fig:additionGadget}.
The pieces $Q_i,Q_j,Q_l$ cover together the restricted region $\Gamma'$ and also the axis-parallel bounding box $R_\Gamma$ of $\Gamma'$.
}
\label{fig:additionGadgetCloseup}
\end{figure}

\subsection{The $\le$-addition gadget}\label{sec:LEadditionGadget}

We obtain the $\leq$-addition gadget by reversing the orientations of all variable segments of the $\geq$-addition gadget.
Instead of placing $r_l$ with the left endpoint at $(-2,-2h)$, we now instead need to place the right endpoint of $r_l$ at the point $(2,-2h)$, so that the scales of the segments $r_i^4,r_j^2,r_l$ match.
Recall from Observation~\ref{obs:gadgetLines} that we want the endpoints of the segments $r_i^1,r_j^1,r_l$ (which are connected to the base pockets via the corridor) to lie on lines with direction $(1,-1)$ with an equidistant horizontal spacing of $13.5/CN^2$.
Since we move $r_l$, we also adjust $r_i^1$ and $r_j^1$ accordingly.
The result can be seen in Figure~\ref{fig:additionGadget2}.

\begin{figure}[h!]
\centering
\includegraphics[width=\textwidth]{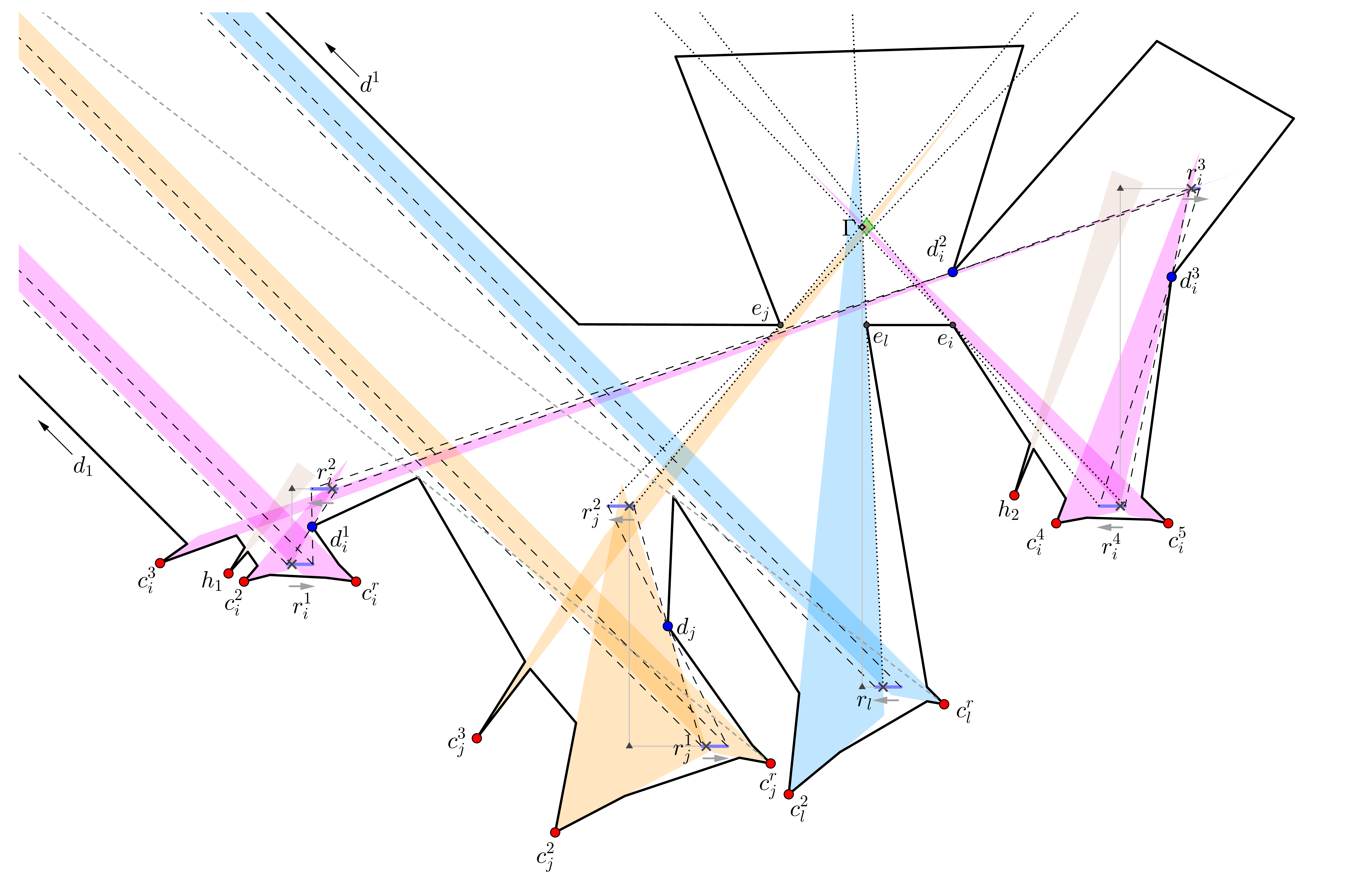}
\caption{The $\leq$-addition gadget.
}
\label{fig:additionGadget2}
\end{figure}

\subsection{The $\ge$-inversion gadget}\label{sec:inversionGadget2}

\begin{figure}[h!]
\centering
\includegraphics[width=0.7\textwidth]{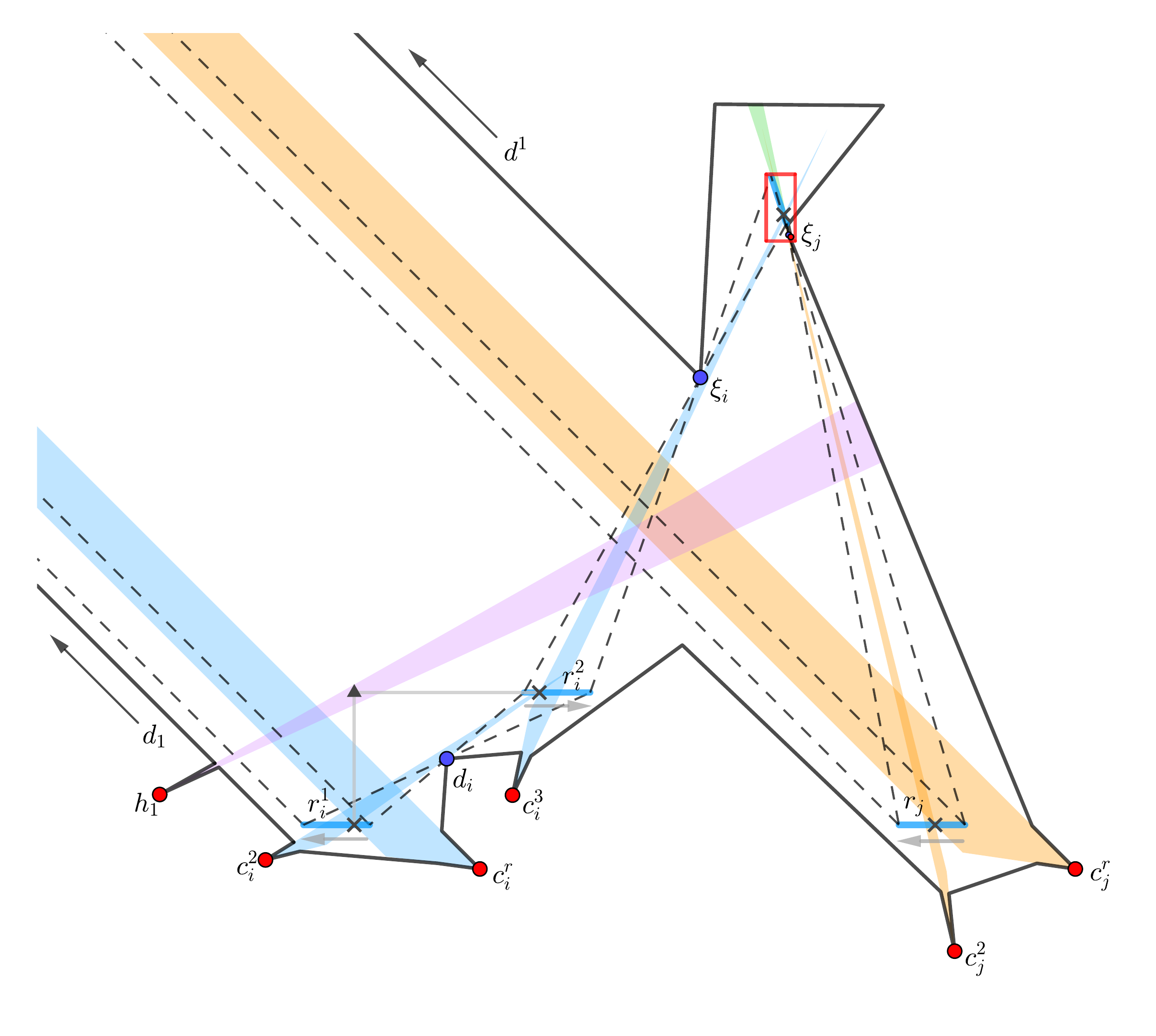}
\includegraphics[width=0.25\textwidth]{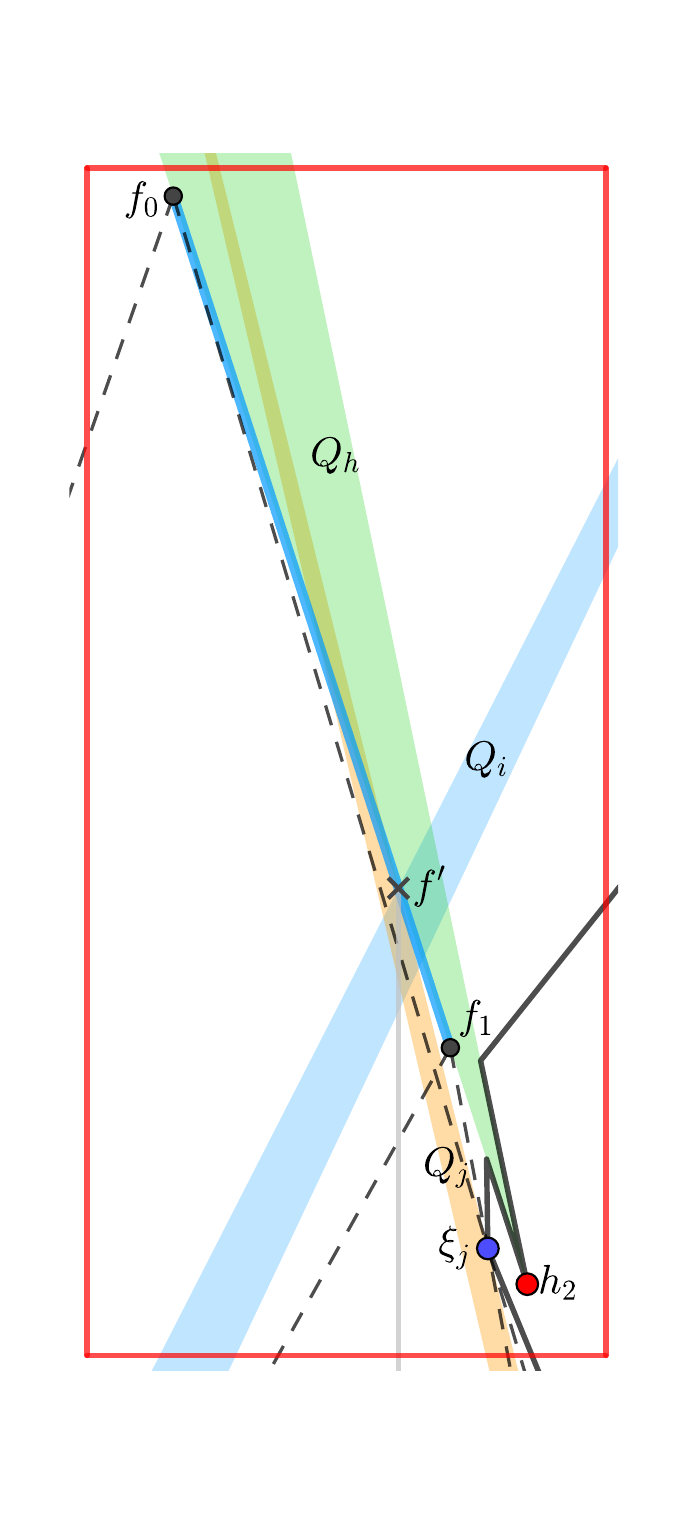}
\caption{The $\geq$-inversion gadget.
To the right is shown a closeup of the region in the red rectangle.
}
\label{fig:inversionGadget2}
\end{figure}

The principle behind the gadget is similar to one from~\cite{abrahamsen2018art}.
The gadget can be seen in Figure~\ref{fig:inversionGadget2}.
The segments $r_i^1$ and $r_j$ are connected to segments in the base pockets via a corridor.
By the use of lever mechanisms, we have
\begin{align*}
\val{r_i}{c_i^r} & \geq \val{r_i}{c_i^2} \geq \val{r_i^2}{c_i^2} \geq \val{r_i^2}{c_i^3} ,\quad\text{and} \\
\val{r_j}{c_j^r} & \leq \val{r_j}{c_j^2}.
\end{align*}

The interesting part happens at a critical segment $f_0f_1$ which $c_i^3$ and $c_j^2$ are responsible for.
The pivots of $c_i^3$ and $c_j^2$ are the corners $\xi_i$ and $\xi_j$, respectively.
The restricted ranges of $r_i^2$ and $r_j$ correspond to an infinitesimal subsegment $f'$ contained in $f_0f_1$, and we can think of $f'$ as a single point.
We have carefully chosen $\xi_i$ and $\xi_j$ such that if the pieces covering $c_i^3$ and $c_j^2$ together cover $f'$, then $\val{r_i^2}{c_i^3} \geq 1/\val{r_j}{c_j^2}$, or equivalently, $\val{r_i^2}{c_i^3}\val{r_j}{c_j^2} \geq 1$.

We now specify the coordinates that make this work.
The described construction should be scaled by a factor of $1/CN^2$; we avoid mentioning this factor for the ease of notation.
The variable segment $r_i^2$ has length $3/2$ and left endpoint $(5.5,3)$, and the segment $r_j$ also has length $3/2$ and left endpoint $(14,0)$.
We now define $\xi_i\mydef (9.5,345/34)\approx (9.5,10.15)$ and $\xi_j\mydef (11.5,51405/3842)\approx (11.5,13.37)$.
This defines a critical segment $f_0f_1$ with $f_0\mydef (15961/1438, 359835/24446)\approx (11.10,14.72)$ and $f_1\mydef (8635/754, 10281/754)\approx (11.45,13.64)$.
We make a helper corner $h_2$ with an incident edge contained in the extension of $f_0f_1$, so that a piece covering $h_2$ can also cover a region above $f_0f_1$.
Pieces covering $c_i^3$, $c_j^2$, and $h_2$ can then together cover the bounding box of the restricted range $f'$ of $f_0f_1$.

Let $\pi_i:r_i^2\longrightarrow f_0f_1$ and $\pi_j: f_0f_1\longrightarrow r_j$ be the projections from $r_i^2$ to $f_0f_1$, and from $f_0f_1$ to $r_j$, respectively.
Consider a point $p\in r_i^2$, which can be written as $p=(5+x,3)$ for $x\in [1/2,2]$, so that $p$ represents the value $x$.
Let $\tilde p \mydef \pi_j(\pi_i(p))$ be the point on $r_j$ corresponding to $p$ on $r_i^2$.
By evaluating the projections, it is straightforward to check that $\tilde p =(16-1/x,0)$.
Hence, $\tilde p$ represents the value $1/x$ at $r_j$.
We now use a similar argument as in Observation~\ref{obs:ineq3}.
Note that a piece covering $c_i^3$ and the point $p$ on $r_i^2$ can cover only a part of $f_0f_1$ from $\pi_i(p)$ to $f_1$.
Hence, the piece $\piece_j$ covering $c_j^2$ must cover the point $\pi_i(p)$ on $f_0f_1$.
It then follows that the rightmost point that $\piece_j$ can cover of $r_j$ is $\pi_j(\pi_i(p))$.
Hence, $\val{r_j}{c_j^2}\geq 1/x$.

Consider now any cover.
We then have $\val{r_i^2}{c_i^3}=x$ and $\val{r_j}{c_j^2}\geq 1/x$, and hence $\val{r_i^2}{c_i^3}\val{r_j}{c_j^2}\geq 1$, as claimed.

The more complicated fractions appear in the coordinates of $\xi_i$ and $\xi_j$ in order to avoid corners with irrational coordinates.
When finding the construction, we made the restriction to coordinates of the forms $\xi_i=(9.5,\chi_i)$ and $\xi_j=(11.5,\chi_j)$.
For any value of $\chi_i$, there is a specific value of $\chi_j$ that makes the right correspondence between $r_i^2$ and $r_j$, but the solution will in general involve a square root.
However, for some exceptional values of $\chi_i$, the interior of the square root is the square of a rational number, and hence the resulting coordinate of $\xi_j$ is rational.
We used Maple to find such a value of $\chi_i$.

\subsection{The $\le$-inversion gadget}\label{sec:inversionGadget1}

\begin{figure}
\centering
\includegraphics[width=\textwidth]{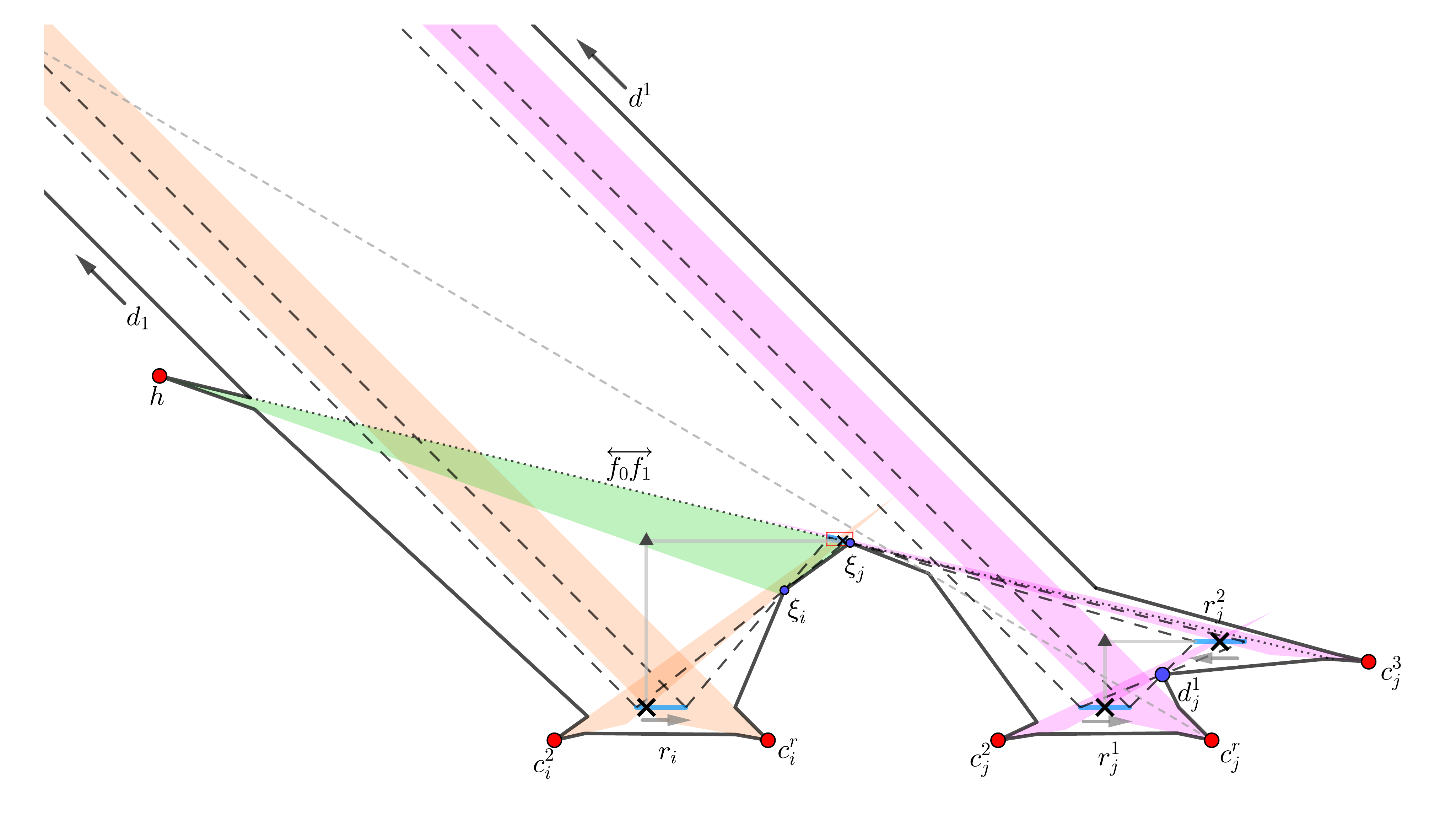}
\caption{The $\geq$-inversion gadget.
For a closeup of the region in the red rectangle, see Figure~\ref{fig:inversionGadgetCloseup}.
}
\label{fig:inversionGadget}
\end{figure}

\begin{figure}
\centering
\includegraphics[width=0.5\textwidth]{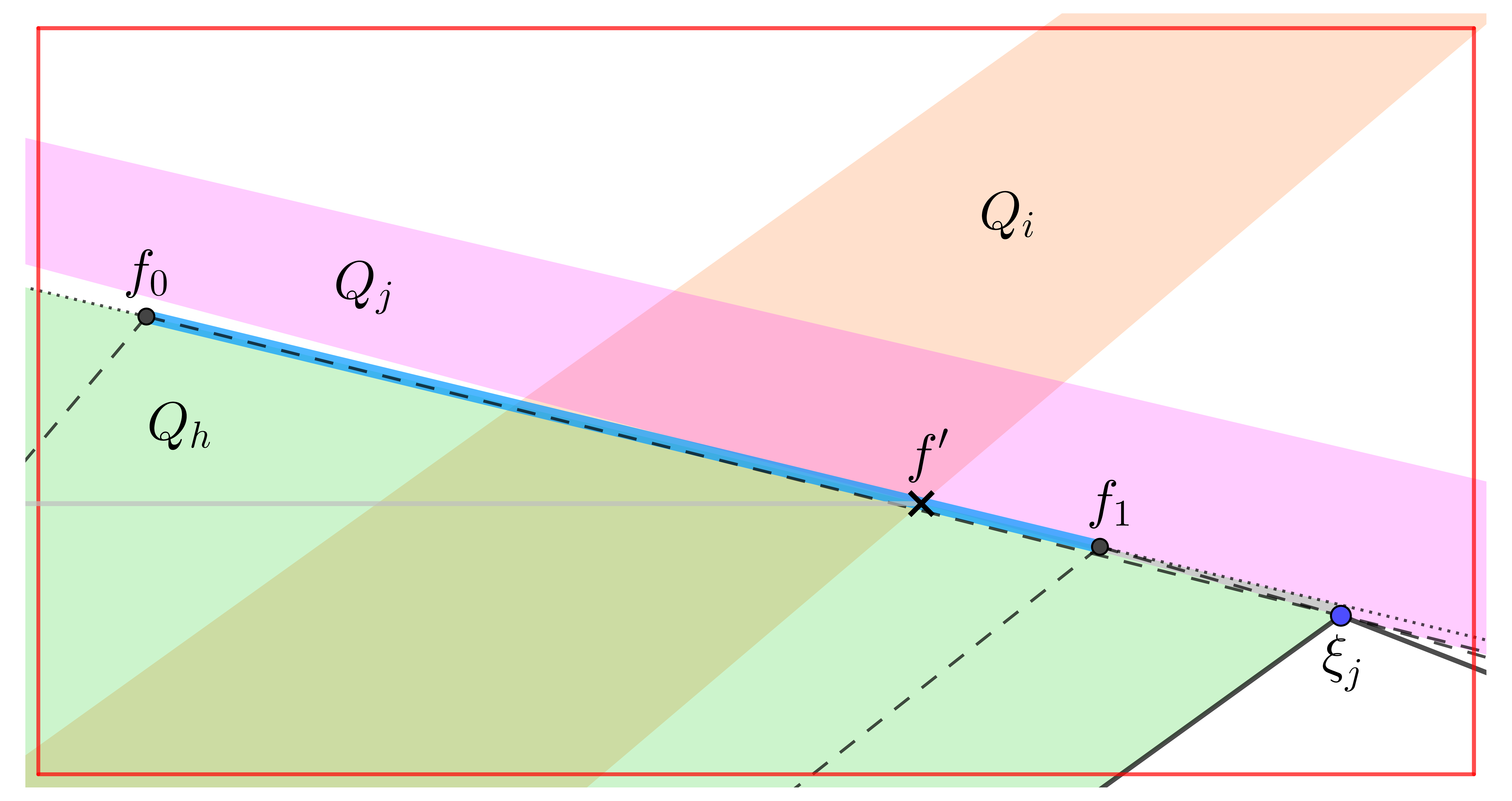}
\caption{Closeup from Figure~\ref{fig:inversionGadget}.
}
\label{fig:inversionGadgetCloseup}
\end{figure}

The $\leq$-inversion gadget follows a similar principle as the $\geq$-inversion gadget, but is in some ways reversed.
The gadget can be seen in Figure~\ref{fig:additionGadget2} with a closeup of an important but small region in Figure~\ref{fig:inversionGadgetCloseup}.
The segments $r_i$ and $r_j^1$ are connected to segments in the base pockets via the corridor.
By the use of lever mechanisms, we have
\begin{align*}
\val{r_i}{c_i^r} & \geq \val{r_i}{c_i^2},\quad\text{and} \\
\val{r_j^1}{c_j^r} & \leq \val{r_j^1}{c_j^2} \leq \val{r_j^2}{c_j^2} \leq \val{r_j^2}{c_j^3}.
\end{align*}

We again have a special critical segment $f_0f_1$ which $c_i^2$ and $c_j^3$ are responsible for.
The pivots of $c_i^2$ and $c_j^3$ are the corners $\xi_i$ and $\xi_j$, respectively.
The restricted ranges of $r_i$ and $r_j^2$ correspond to an infinitesimal subsegment $f'$ contained in $f_0f_1$, and we can think of $f'$ as a single point.
We have carefully chosen $\xi_i$ and $\xi_j$ such that if the pieces covering $c_i^2$ and $c_j^3$ together cover $f'$, then $\val{r_i}{c_i^2} \leq 1/\val{r_j^2}{c_j^3}$, or equivalently, $\val{r_i}{c_i^2}\val{r_j^2}{c_j^3} \leq 1$.

We now specify the coordinates that make this work.
We again leave out a factor of $1/CN^2$ in the coordinates for the ease of notation.
The variable segment $r_i$ has length $3/2$ and left endpoint $(0.5,0)$, and the segment $r_j^2$ also has length $3/2$ and left endpoint $(17.5,2)$.
We now define $\xi_i\mydef (5, 57/16)\approx (5, 3.56)$ and $\xi_j\mydef (7, 1041/208)\approx (7,5)$.
This defines a critical segment $f_0f_1$ with $f_0\mydef (235/37, 3059/592)\approx (6.35, 5.17)$ and $f_1\mydef (735/107, 25897/5136)\approx (6.87, 5.04)$.
We make a helper corner $h$ with an incident edge contained in the extension of $f_0f_1$, so that a piece covering $h$ can also cover a region below $f_0f_1$.
Pieces covering $c_i^2$, $c_j^3$, and $h$ can then together cover the bounding box of the restricted range $f'$ of $f_0f_1$.
A piece covering $h$ can likewise cover a marked rectangle that is introduced in order to satisfy the bar intersection promise.

As in the $\geq$-inequality gadget, it now follows that $\val{r_i^2}{c_i^3}\val{r_j}{c_j^2}\leq 1$.

\subsection{Putting it all together}\label{sec:putting}

Let us now summarize the entire construction.
We consider an instance $\II_1$ of \rangeetrinv\ with $n$ variables and a formula $\Phi$ consisting of addition and inversion constraints.
We rewrite $\Phi$ as a formula $\Phiineq$, which is a conjunction of inequalities among $3n$ variables $\Xineq\mydef\{x_1,\ldots,x_{3n}\}$, as described in Section~\ref{sec:bottomWall}.
We make a base pocket for each variable $x_i\in \Xineq$ and a variable segment in this base pocket for each occurrence of $x_i$ in an inequality in $\Phiineq$.
We then make a gadget for each inequality, as described in Sections~\ref{sec:Geinequality}--\ref{sec:inversionGadget1}.
The gadgets are connected to the base pockets by corridors, as described in Section~\ref{sec:copy}, which ensures that the values represented in the gadgets satisfy appropriate inequalities compared to those in the base pockets.

Suppose that the constructed instance $\II_2$ has a cover.
By the correctness of the corridor, Lemma~\ref{lem:corridor_works}, we have inequalities between the values in the base pockets and in the gadgets as intended.
As described in Section~\ref{sec:consistency}, the use of the $\ge$- and $\le$-inequality gadgets ensure that the variables in the base pockets are represented consistently.
By the correctness of the addition and inversion gadgets, Sections~\ref{sec:GEadditionGadget}--\ref{sec:inversionGadget1}, we then get that the values specified on the variable segments correspond to a solution to $\Phi$.

On the other hand, given a solution to $\Phi$, we get a corresponding cover for $\II_2$ using only maximal triangles in the following way.
For each variable segment $s$ and each lever corner $c$ responsible for $s$, we use a triangular piece with a corner at $c$ and two edges that are extensions of the edges of $\poly$ incident at $c$.
The last edge is chosen so that at $s$, the triangle represents the value specified by the solution to $\Phi$.
These triangles together cover all variable segments and also the remaining (non-horizontal) critical segments.
By adding triangles that cover the helper corners, we can make a cover for all marked rectangles.
It is also clear that the instance $\II_2$ can be constructed from $\II_1$ in polynomial time.
We have therefore proved Theorem~\ref{thm:mainthm} and in particular that \mrcc\ is $\ER$-complete.

In order to use Lemma~\ref{thm:mrcc-to-mcc} to prove that the original problem \mcc\ is also $\ER$-complete, we need the constructed instance $\II_2$ to satisfy the promise of \mrcc-instances, described in Section~\ref{sec:mrcc}.
It is straightforward to verify the skew triangle promise.
Regarding the trapezoid generality promise, it is sufficient that no two corners are visible and have the same $x$- or $y$-coordinate.
This is not true for the construction as described:
The corners next to the propagation corners in a base pocket are mutually visible and on a vertical line.
Likewise are the corners of the left doors of the corridors.
However, it is easy to see that all these corners can be slightly moved, say, by a distance of $\Theta(1/N^3)$, without affecting the correctness of the construction, and then the promise will be satisfied.

For the bar intersection promise, we note that a bar from marked rectangles in a base pocket cannot intersect bars from other base pockets or corridors or gadgets.
Likewise, bars in corridors cannot intersect bars from other corridors, base pockets or gadgets.
In the sections describing the individual components (base pockets, corridors, gadgets), we have been careful to satisfy the bar intersection promise within that component.
It is therefore also satisfied globally.
Finally, for the broad cover property, suppose that $\II_2$ has a cover $\cover$.
We then have a solution to $\Phi$, as described above.
This solution corresponds to a cover $\cover'$ for $\II_2$ using only triangles, and for each marked corner $c$, the triangle $T$ covering $c$ has edges that are extensions of the edges of $\poly$ incident at $c$.
Hence, $\Delta(c)\subset T$, as needed.

Since we have verified that the instance $\II_2$ can be assumed to satisfy the promise, we then get from Lemma~\ref{thm:mrcc-to-mcc} that \mcc\ is $\ER$-complete.
We now observe that if there exists a cover of the resulting instance $\II'_2$ of \mcc, then there also exists a cover consisting of triangles only:
As we have seen before, there exists a triangle cover for the instance $\II_2$ of \mrcc.
It is described in the proof of Lemma~\ref{thm:mrcc-to-mcc} how we can then get a cover for $\II_2'$ by covering the added spikes and remaining parts of the polygon with triangles.
Therefore, we also get $\ER$-completeness of \mtcf.
We have then proved Theorem~\ref{thm:mainThm}, repeated here. 

\mainThm*

We can now also prove Corollary~\ref{corr:corr}.

\corr*

\begin{proof}
Let $\Phi=\Phi(x)$ be the \etr-formula $x\cdot x\cdot x\cdot x\cdot x-x-x-x-x+1+1=0$, i.e., the equation $x^5-4x+2=0$ written up in the way required for an \etr-formula.
As mentioned, the equation is known to have a single real solution $x=y$ for a number $y\in\mathbb R$ which is not expressible by radicals~\cite{bastida_lyndon_1984}.
We reduce $\Phi$ to an instance $\II_1$ of \rangeetrinv\ using the reduction from~\cite{abrahamsen2020framework}, which can then be reduced to an instance $\langle \poly,k\rangle$ of \mcc\ using the reductions described in this paper.
By looking at the reduction from~\cite{abrahamsen2020framework}, one sees that one of the variables $x'$ in $\II_1$ corresponds to $x$ in the formula $\Phi$, in the sense that $x'=q_1x+q_2$ for two fixed rationals $q_1,q_2\in\mathbb Q$.
Hence, in all solutions to $\II_1$, $x'$ has a value $y'\in\mathbb R$ that is not expressible by radicals.
Consider a variable segment $s$ in $\poly$ that represents $x'$, and let $c$ be a lever corner responsible for $s$ and $\piece$ a piece covering $c$ in a cover for $\poly$.
We then have $\val sc=y'$.
Let $e_1e_2$ be the edge of $\piece$ that intersects $s$ and thus defines this value $y'$.
If the corners $e_1$ and $e_2$ of $\piece$ had coordinates expressible by radicals, then so would the intersection point of $e_1e_2$ with $s$, as the endpoints of $s$ are rational, and then the represented value would be expressible by radicals as well.
Therefore, $e_1$ or $e_2$ has a coordinate that is not expressible by radicals.
\end{proof}

\section{Concluding remarks}\label{sec:concluding}

As described, the polygon produced by our reduction has some triples of collinear corners.
We here outline how it can be changed to avoid this.
The proof of Lemma~\ref{thm:mrcc-to-mcc} describes the transformation from an instance $\II_2$ of \mrcc\ to an instance $\II'_2$ of \mcc\ by adding some spikes to the polygon.
The spikes can be covered by one triangle each, such that these triangles also cover everything outside the marked rectangles and the triangles $\Delta(c)$ for each marked corner $c$.
The spikes introduce some collinear points of two types:
(i) The way we originally drew the spikes in Figures~\ref{fig:spikes2} and~\ref{fig:spikesMove}, two different spikes may have edges contained in the same line, creating four collinear corners.
(ii) Additionally, some spikes may have a horizontal or vertical edge that also contains a corner of the original polygon.

In order to eliminate these collinear corners, we do as follows.
For each critical segment $s$, we consider the marked rectangle $R$ containing $s$.
We now make sure that $s$ is disjoint from the boundary of $R$.
If $s$ intersects the boundary $\partial R$, we make $R$ a tiny bit larger so that $s$ is contained in the interior; see Figure~\ref{fig:newhelper}.
If $s$ is a variable segment, then $s$ would originally be contained in the top or bottom edge of $R$, so when enlarging $R$, we also need to introduce a new helper corner that can cover the part of $R$ on one side of $s$, as the pieces covering the corners responsible for $s$ are only covering the other side of $s$.
With this modification, we again obtain that the instance with enlarged rectangles has a cover if and only if $\Phi$ has a solution.

\begin{figure}
\centering
\includegraphics[page=28]{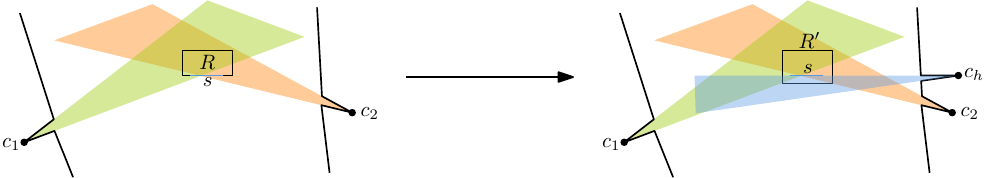}
\caption{Instead of the marked rectangle $R$ where $s$ is contained in the bottom segment, we use the slightly larger $R'$ where $s$ is in the interior.
We then add the helper spike with a corner $c_h$ that can cover the part of $R'$ below $s$.
}
\label{fig:newhelper}
\end{figure}

When reducing to \mcc\ and creating the instance $\II'_2$, we now enlarge the spikes a bit, as shown in Figure~\ref{fig:generalSpikes}, which eliminates both collinearities of types (i) and (ii).
We have kept the property that no piece in a cover for $\II_2'$ can cover two spikes, so there must be a piece for each spike in a cover for $\II_2'$.
It is now conceivable that a piece covering a spike can also cover a part of a marked rectangle, but by changing the spikes sufficiently little, the covered part of the rectangle is very close to the boundary, leaving the critical segment inside to be covered by the responsible marked corners.
Hence, the reduction still works as intended.

\begin{figure}
\centering
\includegraphics[page=27]{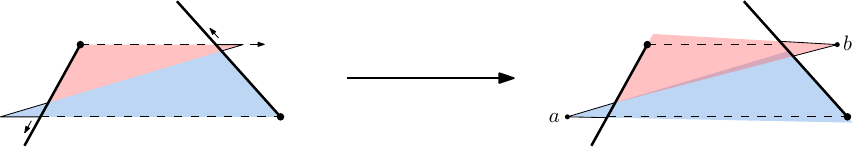}
\caption{Enlarging spikes to avoid collinear points.
Note that the resulting spike corners $a$ and $b$ do not see each other.
}
\label{fig:generalSpikes}
\end{figure}

An interesting question that remains open is to determine the complexity of covering a simple polygon with a minimum number of spiral polygons.

\section*{Acknowledgements}
I thank Aurélien Ooms for useful discussions in the initial phases of this work, and Joseph O'Rourke for sending me his paper~\cite{o1982complexity}.

\printbibliography

\end{document}